\documentclass[
  aps,
  reprint,
  superscriptaddress,
  longbibliography,
  amssymb
]{revtex4-2}

\usepackage[dvipsnames]{xcolor}
\usepackage{graphicx}
\usepackage{braket}
\usepackage{amsmath}
\usepackage{amsthm}
\usepackage{times}

\usepackage[
  colorlinks=true,
  hypertexnames=false
  ]{hyperref}

\PassOptionsToPackage{linktocpage}{hyperref} 

\hypersetup{
  colorlinks   = true, 
  urlcolor = magenta!80!black,
  linkcolor    = blue!60!black, 
  citecolor=black!60 
}
\usepackage[capitalize]{cleveref}

\usepackage{dsfont}  
\usepackage[caption = false]{subfig}

\newtheorem{thm}{Theorem}
\newtheorem{definition}{Definition}
\newtheorem{lemma}{Lemma}

\newcommand{\C}{\mathbb{C}}
\newcommand{\diag}{\mathrm{diag}}

\newcommand{\ee}{\mathrm{e}}
\newcommand{\ii}{\mathrm{i}}
\newcommand{\dd}{\mathrm{d}}
\newcommand{\rr}{\mathbb{R}}

\newcommand{\pp}{\mathbb{P}}

\newcommand{\I}{\mathds{1}}
\newcommand{\loss}{\mathcal{L}}
\newcommand{\D}{\mathcal{D}}
\newcommand{\dg}[1]{{#1}^\dagger}
\newcommand{\dt}{\Delta t}

\newcommand{\quics}{Joint Center for Quantum Information and Computer Science (QuICS), University of Maryland \& NIST, College Park, MD 20742, USA}

\newcommand{\jqi}{Joint Quantum Institute (JQI), University of Maryland \& NIST, College Park, MD 20742, USA}
\begin{document}

\title{Scalably learning quantum many-body Hamiltonians from dynamical data}
\author{F.~Wilde}
\affiliation{Dahlem Center for Complex Quantum Systems, Freie Universit\"{a}t Berlin, 
14195 Berlin, Germany}
\author{A.~Kshetrimayum}
\affiliation{Helmholtz-Zentrum Berlin f{\"u}r Materialien und Energie, 14109 Berlin, Germany}
\affiliation{Dahlem Center for Complex Quantum Systems, Freie Universit\"{a}t Berlin, 
14195 Berlin, Germany}
\email{this is a test}
\author{I.~Roth} 
\affiliation{Quantum Research Centre, Technology Innovation Institute (TII), Abu Dhabi}
\author{D.~Hangleiter} 
\affiliation{\quics}
\affiliation{\jqi}

\author{R.~Sweke}
\altaffiliation[Currently at ]{IBM Quantum, Almaden Research Center, San Jose, CA 95120, USA}
\affiliation{Dahlem Center for Complex Quantum Systems, Freie Universit\"{a}t Berlin, 
14195 Berlin, Germany}
\author{J.~Eisert}
\affiliation{Dahlem Center for Complex Quantum Systems, Freie Universit\"{a}t Berlin, 
14195 Berlin, Germany}
\affiliation{Helmholtz-Zentrum Berlin f{\"u}r Materialien und Energie, 14109 Berlin, Germany}
\affiliation{Fraunhofer Heinrich Hertz Institute, 10587 Berlin, Germany}

\date{\today}
\begin{abstract}
The physics of a closed quantum mechanical system is governed by its Hamiltonian. However, in most practical situations, this Hamiltonian is not precisely known, and ultimately all there is are data obtained from measurements on the system. In this work, we introduce a highly scalable, data-driven approach to learning families of interacting many-body Hamiltonians from dynamical data, by bringing together techniques from gradient-based optimization from machine learning with efficient quantum state representations in terms of tensor networks. Our approach is highly practical, experimentally friendly, and intrinsically scalable to allow for system sizes of above 100 spins. In particular, we demonstrate on synthetic data that the algorithm works even if one is restricted to one simple initial state, a small number of single-qubit observables, and time evolution up to relatively short times. For the concrete example of the one-dimensional Heisenberg model our algorithm exhibits an error constant in the system size and scaling as the inverse square root of the size of the data set.
\end{abstract}
\maketitle

Given the Hamiltonian of a quantum system we can, in principle, derive all physical properties of that system. Therefore, many studies of theoretical physics start by specifying the Hamiltonian. 
In practice, an effective Hamiltonian is typically hypothesised based on theoretical considerations. 
This hypothesis is then confirmed or rejected by comparing its predictions with experimental data.
However, in many contexts, it is far from clear what the effective Hamiltonian describing a physical system actually is, at least to high levels of accuracy.
In this situation, one might hope to infer the Hamiltonian governing the system directly from experimental data or, in other words, to \emph{learn the Hamiltonian}.

In addition to their conceptual, foundational significance, there is a second more technologically minded reason why Hamiltonian learning is important:
In the context of quantum information science, one aims at making predictions to high and often unprecedented accuracy.
Analog quantum simulators in particular allow to assess the physics of interacting quantum systems presumably beyond the reach of classical computers~\cite{CiracZollerSimulation, Trotzky,2dMBLScience}.
To build trust in such devices it is neccessary to precisely determine the Hamiltonian at play.
Moreover, in the context of near-term quantum computing, engineering and testing quantum processors requires precise knowledge of the underlying physics \cite{Carrasco_Verification_2021}.
Hamiltonian learning can thus be seen as a primitive in the precise engineering of quantum devices.
As a first step toward the ultimate goal of learning an effective Hamiltonian from scratch, one might consider the less demanding task of learning the precise parameters of a given Hamiltonian model.

In this work, we introduce a technique that allows to scalably learn the Hamiltonian parameters of interacting quantum many-body systems.
To achieve applicability of our method to large systems, we bring together two computational methodologies.
On the one hand, we make use of \emph{tensor network techniques} that allow to efficiently compute the properties of a large class of quantum many-body systems in space and---at least for sufficiently small times---in time.
On the other hand, we get inspiration from \emph{machine-learning} approaches and make use of stochastic gradient descent optimization methods~\cite{krastanov_stochastic_2019} and automatic differentiation~\cite{liao_differentiable_2019}.
On a high level, we solve a \emph{maximum likelihood estimation} (MLE) problem on large sets of dynamical measurement data using (stochastic) gradient-based optimization.
The data we use are simply given by the outcomes of Pauli measurements of time-evolved states starting from one simple reference state.

\begin{figure*}
  \centering
  \includegraphics[width=.85\textwidth]{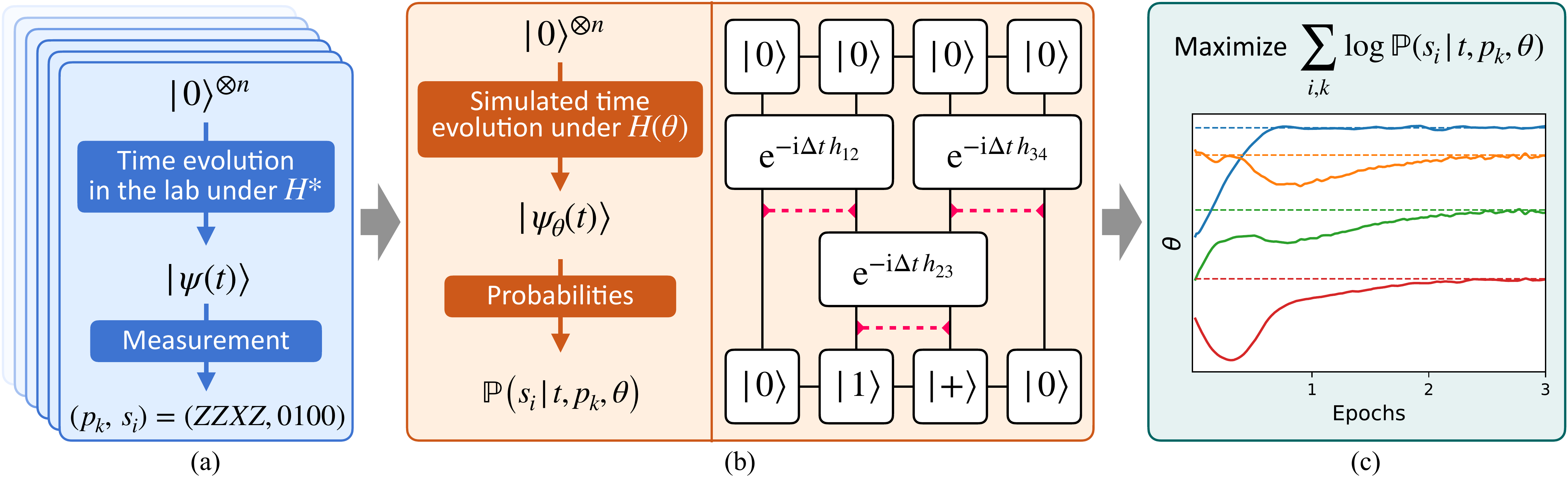}
  \caption{\label{fig:setting}
    The core idea of the method.
    \textbf{(a)}~Data are accumulated from repeated measurements at different times.
    Each data point corresponds to a randomly chosen Pauli basis and a string of binary measurement outcomes.
    \textbf{(b)}~The time evolution is simulated under a Hamiltonian $H(\theta)\in {\cal C}$ using the TEBD algorithm.
    After each contraction with a local time step $\ee^{-\ii \dt\,h_{\alpha,\beta}}$ the resulting rank-4 tensor is split and the bond dimensions are truncated to their original size, as indicated by the dashed red lines.
    \textbf{(c)}~The best suitable parameters $\theta^\mathcal{D}\in\rr^\nu$---according to the negative log-likelihood cost function---matching the observed data are learned using automatic differentiation and gradient-based optimization.
    In the case shown here $\nu=4$ and the true parameter values (dashed lines) are recovered successfully.
  }
\end{figure*}

Previous work on many-body Hamiltonian learning \footnote{See Ref.~\cite{Carrasco_Verification_2021} for a recent perspective article.} can be roughly divided into two categories. 
On the one hand, there is a body of theoretical work which rigorously shows that many-body Hamiltonians can be efficiently learned from either thermal states or eigenstates of the Hamiltonian \cite{bairey_learning_2019,BaireyNJP,anshu_sample-efficient_2021,haah_optimal_2021,qi_determining_2019}, or short-time evolution \cite{zubida_optimized_2021,franca_efficient_2022,gu_practical_2022,harper_fast_2021}, sometimes even in a way that is robust to \emph{state preparation and measurement} (SPAM) errors \cite{yu_practical_2022}.
On the other hand, there are practically motivated approaches to learn Hamiltonians that make use of the entire long-time dynamics of a system. 
Such methods face the challenge of predicting the time evolution. 
To overcome this challenge, 
one can utilize specific structure of certain classes of Hamiltonians~\cite{oi_quantum_2012,burgarth_indirect_2011,
GoogleHamiltonianLearningShort},
algebraic properties of the Hamiltonian \cite{zhang_quantum_2014,chen_experimental_2021}, or their conserved quantities \cite{li_hamiltonian_2020,pastori_characterization_2022}. 
Alternatively, one can also make use of high-level machine-learning techniques such as neural networks~\cite{bienias_meta_2021, schuster_learning_2022,
mohseni_deep_2021, han_tomography_2021, valenti_scalable_2022} to predict the time series. 
Close in mindset to our approach are schemes such as the one of Ref.~\cite{xie_bayesian_2022} which makes use of tensor-network techniques to directly learn a control model of an analogue quantum simulator.

However, few steps have been taken so far to devise practical methods for learning interacting quantum many-body Hamiltonians of large systems with a large number of unknown parameters.
A large-scale demonstration on up to 100 sites with many parameters has only been achieved using steady states of the Hamiltonian~\cite{evans_scalable_2019}, while methods that are in principle scalable and practical have not yet been demonstrated for large system sizes~\cite{li_hamiltonian_2020}, or have prohibitively large (albeit polynomial) scaling~\cite{zubida_optimized_2021,yu_practical_2022}.
Moreover, such methods tend to require very specific state preparations and/or measurements, as well as the estimation of expectation values. 

In contrast, our learning algorithm directly runs on the measurement outcomes of a small number of Pauli measurements, which can be chosen suitably depending on the measurement setup. 
This ensures that we do not discard any correlations contained in the measurement, and reduces the total number of experiments required. 
With sufficient computational memory, the run-time of our method scales linearly in the system size (see Fig.~\ref{fig:timing}) and the reconstruction error scales as an inverse square root of the total number of measurements.

\paragraph*{A scalable method for Hamiltonian learning.}
We consider the following setting, as illustrated in \cref{fig:setting}: A closed quantum many-body 
system consisting of $n$ spins---or qubits---in the laboratory, governed by the Hamiltonian $H^*$, is initialized in a low-entanglement state.
Specifically, here we consider the initial state to be the product state vector $\ket{0}^{\otimes n}$.
Then the system evolves in time according to the Schr\"odinger equation up to some time $t_j$ whereupon all spins are measured in some, potentially randomly chosen, Pauli basis $p_k \in \lbrace X, Y, Z\rbrace^n$.
The measurement outcome is a binary string $s$ of length $n$.
In general this process is repeated $M$ times for multiple time stamps $t_1, \ldots, t_J$ and Pauli bases $p_1, \ldots, p_K$.
Hence, the resulting data set $\D$ contains $M\!\cdot\!J\!\cdot\!K$ bit-strings.
Alternatively, one could prepare $K$ initial states and only measure in one Pauli basis.

To simulate the laboratory system, we make use of tensor networks, which is a powerful tool for describing quantum many-body systems beyond the reach of exact diagonalization~\cite{schollwock_density-matrix_2011}.
In particular, we use the \emph{time-evolving block decimation} (TEBD) algorithm, illustrated in Fig.~\ref{fig:setting}(b), which is a particularly simple and efficient tensor network algorithm that can be used for computing ground states, thermal states, and non-equilibrium dynamics of quantum many-body systems~\cite{Vidaltebd,ZwolakPRL2004}.

In this work, we focus on the dynamics of one-dimensional systems, although an extension to higher dimensions is also possible~\cite{Kshetrimayum2dMBL2020,Hubigpepsdynamics,Kshetrimayum2dTC2021}.
By choosing an initial state with low entanglement, we can represent it using a \emph{matrix-product state} (MPS) efficiently in space until intermediate times.
To solve the Schr\"odinger equation, the time-evolution operator $\ee^{-\ii tH}$ is decomposed into a product of small Trotter steps $\ee^{-\ii\dt H} \cdots \ee^{-\ii\dt H}$.
Additionally, the operator is decomposed into layers of mutually commmuting terms, as shown in the example in 
Fig.~\ref{fig:setting}(b).
These decompositions incur an approximation error, called \emph{Trotter error}, which can be controlled by the Trotter-step size $\dt$.
Each time, after contracting a layer of operators with the MPS, the tensors are split using a \emph{singular value decomposition} (SVD) to retain the MPS structure with a given bond dimension.
After each splitting the bond dimension between tensors gets multiplied by the physical dimension of the sites, which in our case is 2.
To mitigate this exponential growth the bond dimension is truncated after splitting, which incurs the \emph{truncation error}.
Generally, the longer the time evolution is, the more entanglement is created in the system, which increases the truncation error.
As such, TEBD is limited to intermediate times with a constant scaling in the system size, but importantly, not as short that expansions in powers of the Hamiltonian are suitable.
Note that the computational complexity of TEBD scales linearly with the system size in one dimension.
Hence, when the set of timestamps $\lbrace t_j\rbrace$ and the size $d$ of the data set is constant, the computational cost of one gradient evaluation also scales linearly, due to the use of automatic differentiation, as shown in detail in Appendix~\ref{sec:implementation}.

To recover $H^*$, we fix an ansatz class of Hamiltonians governed by certain Hamiltonian parameters. 
Formally, we define a \emph{class of parametrized Hamiltonians} $\mathcal{C} = \lbrace H(\theta) \vert \theta \in \Theta \rbrace$ as the range of a twice continuously differentiable function $\theta \mapsto H(\theta)$ (see Appendix \ref{sec:normality}).
The parametrization is chosen according to physical intuition.
For instance, the parameters could be the amplitudes of interactions one expects to occur in the lab experiment.
Then, we classically simulate the time evolution under the parametrized Hamiltonian $H(\theta_\mathrm{ini})\in {\cal C}$ using TEBD, where $\theta_\mathrm{ini}$ is a, typically random, initialization of the parameters.
Based on the data $\D$ we formulate a MLE problem by defining the \emph{negative log-likelihood}
\begin{equation}
  \label{eq:loss}
  \loss^\mathcal{D}(\theta) = - \frac{1}{d}\sum_{j=1}^J\sum_{k=1}^K\sum_{i=1}^M \log \pp(s_{i,j,k} \vert t_j, p_k, \theta),
\end{equation}
with $d = \vert \D\vert$, which plays the role of a \emph{loss function}.
The minimizer $\theta^\mathcal{D}=\mathrm{argmin}\,\loss^\mathcal{D}(\theta)$ of the loss corresponds to the Hamiltonian in $\mathcal{C}$ which best describes the given data.
Note that due to this approach there is no need to post-process the data, for instance, by estimating single-spin expectation values.
Therefore, all correlations contained in the data set are accessible to the learning algorithm.
To evaluate the loss function we compute the probabilities according to the Born rule $\pp(s_{i,j,k} \vert t_j, p_k, \theta) = \vert\!\bra{\phi_{i,j,k}}\ee^{-\ii H(\theta)t_j}\ket{0}\!\vert^2$, where $\ket{\phi_{i,j,k}}$ is the $i$-th post-measurement state vector corresponding to the $p_k$~measurement after evolution up to time $t_j$.
This probability can be computed using TEBD as illustrated by the tensor-network diagram in Fig.~\ref{fig:setting}(b).

When the number of parameters $\nu$ is large, i.e., the minimization problem is high dimensional, gradient-based optimization becomes advantageous over zeroth-order methods.
In order to obtain the gradient efficiently and accurately, even for very large $\nu$, we use automatic differentiation to back-propagate the derivative through the entire TEBD algorithm.
This has been implemented using the \emph{auto-differentiation framework} JAX~\cite{jax2018github}, which supports complex numbers and matrix operations such as the matrix exponential and the singular value decomposition.
To allow the differentiation of the SVD with complex inputs we derived and implemented the Jacobian-vector product rule and integrated it into JAX (see Appendix~\ref{sec:svd-ad}).
Additionally, we extensively exploit the potential of the Accelerated Linear Algebra (XLA) compiler which JAX is built on.
Specifically, we use just-in-time compilation and vectorization by writing the TEBD using only XLA-compatible control flow syntax.
Further implementation details can be found in Appendix~\ref{sec:implementation} and the Github repository~\cite{github_dynamical-hamiltonian-learning_nodate}.

The algorithm we use to minimize $\loss$ consists of two stages.
In the first stage we use the popular optimizer ADAM and mini-batch stochastic gradient descent.
In each step of the descent we select one bit-string per Pauli basis, while the sets of bit-strings get shuffled for each epoch.
The resulting batch of data, containing $K\!\cdot\!J$ bit-strings, is then used to compute the stochastic gradient estimator.
As this reduces the computation time per gradient step drastically, the first stage quickly moves to the vicinity of a local minimum.
To find the minimum with high precision, in the second stage we use the entire data set to compute the gradient and employ the \emph{pseudo Newton optimizer BFGS}.
The gradient steps in the second stage require more computation time, but the BFGS algorithm takes only a few tens of steps to find the minimum with high precision.
For a motivation of this particular optimization algorithm, see Appendix~\ref{sec:implementation}.

\begin{figure*}
  \centering
  \includegraphics[width=\textwidth]{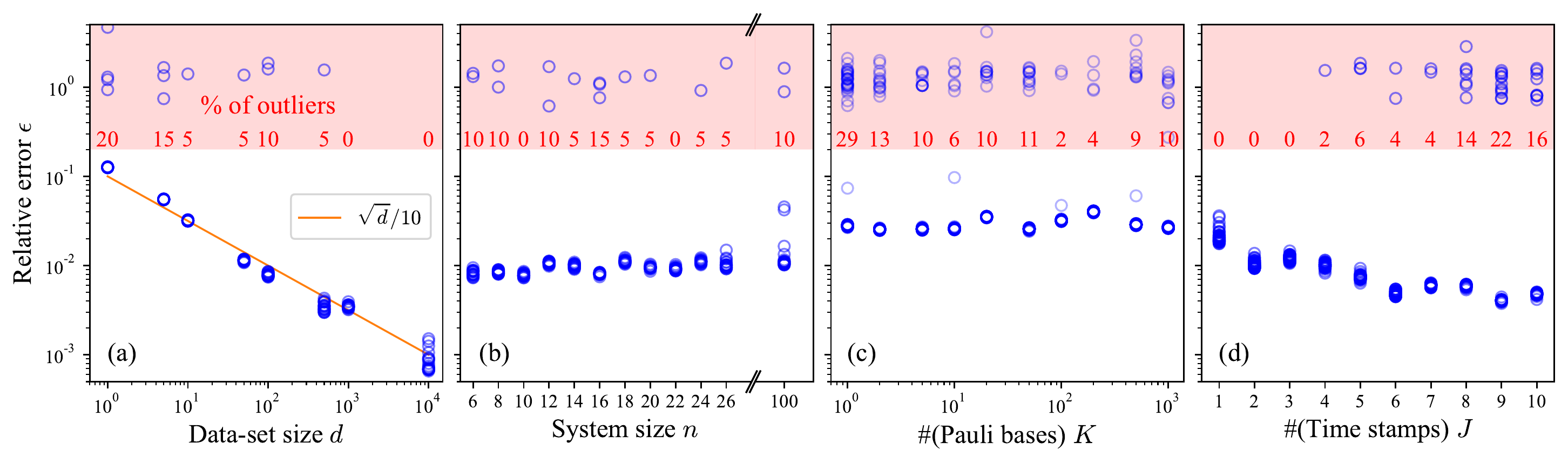}
  \caption{
    \label{fig:results}
    The scaling of the relative recovery error $\epsilon$ for various experimental parameters.
    Each data point (blue circle) corresponds to a unique parameter initialization $\theta_\mathrm{ini}$, drawn at random.
    Vertically aligned data points correspond to the same data set $\mathcal{D}$.
    All points above the threshold $\epsilon=0.2$ (red dashed line) are counted as outliers (red numbers) and deemed as \emph{not successfully converged}.
    In all plots, except (c), the number of Pauli bases is 100.
    In all plots, except (d), the number of time stamps is 5.
    \textbf{(a)}~The error decreases with the size of the data set.
    \textbf{(b)}~The error does not appear to significantly depend on the system size and thereby also not on the number of parameters.
    \textbf{(c)}~Even with only one Pauli basis, we can recover the Hamiltonian, however with a larger amount of outliers.
    The total number of samples has been kept fixed at $5 \cdot 10^3$.
    \textbf{(d)}~The error decreases as one increases the number of time stamps $J$.
    One time stamp means measurements were only taken at time $\tau$, two time stamps implies measurements were taken at times $\tau$ and $2\tau$, etc.
    The total number of measurements has been kept fixed to $\approx 6 \cdot 10^4$ in each column.
    The number of vertically aligned points, i.e., the number of randomly drawn parameter initializations, is 20 for (a) and (b), 100 for (c), and 50 for (d).
  }
\end{figure*}

\paragraph*{Analysis.}
In practice, when the Hamiltonian parameters are unknown, the only accessible indicator of success is the loss function.
Once the loss function has converged to a low value, the target Hamiltonian $H^*$ might have been recovered.
We find that one can distinguish successfully converged instances from instances which converged to a local minimum a posteriori only by the value of their loss function, as we elaborate in \cref{sec:landscape}.
However, in this work, the data are generated synthetically and thus we exactly know the target Hamiltonian.
This allows us to precisely gauge the performance of our method by defining the relative error
\begin{align}
  \label{eq:error}
  \epsilon(\theta) := \frac{\big\Vert\theta - \theta^*\big\Vert_2}{\big\Vert\theta^*\big\Vert_2}.
\end{align}
The data set $\D$ consists of $d$ randomly sampled bit-strings, representing the outcomes of quantum measurements.
Thus it is natural to define the minimizer of the loss function as a random variable $\hat{\theta}^{(d)}$ itself, where the data set size $d$ is an independent parameter and the sets of measurement bases $\lbrace p_k\rbrace$ and times $\lbrace t_j\rbrace$ are fixed.
Consequently, the relative error is also a random variable $\hat{\epsilon}^{(d)} = \epsilon(\hat{\theta}^{(d)})$.
By examining the MLE problem at hand we find that the relative error $\hat{\epsilon}^{(d)}$ scales as $d^{-1/2}$ in the large-sample limit.
\begin{thm}[Asymptotic error (informal)]
  \label{thm:error-informal}
  Suppose the class of parameterized Hamiltonians $\mathcal{C}$ is well conditioned, such that the estimator $\hat{\theta}^{(d)}$ is consistent and the Hessian of the loss is non-singular.
  Then for any $\delta \in (0,1]$ there exists a function $f(d) = \mathcal{O}(d^{-1/2})$ such that $\pp\big[\hat{\epsilon}^{(d)} > f(d)\big] < \delta$ when $d$ is sufficiently large.
\end{thm}
A rigorous statement of the theorem as well as a proof is provided in Appendix~\ref{sec:normality}.
Note that this scaling can be seen in Fig.~\ref{fig:results}(a) even for moderate values of $d$, despite the fact that we cannot prove that our choice of $\mathcal{C}$ is well conditioned.
Lastly, note that the theorem makes a statement about the global minimum $\theta^\mathcal{D}$ of the loss function.
However, it does not guarantee the convergence to $\theta^\mathcal{D}$.

\paragraph*{Results.}
To numerically analyze our method we generated data in two different ways.
For systems consisting of up to 26 spins, the Schr\"odinger equation has been solved via exact state-vector simulation
\footnote{We use the Runge-Kutta ODE solver implemented in \texttt{scipy.integrate.solve\_ivp} to integrate the Schr{\"o}dinger equation in its exact form on the Hilbert space of dimension $2^n$.}.
To go beyond the reach of exact diagonalization, the Schr\"odinger equation has been solved using TEBD for 100 spins with a sufficiently large bond dimension.
Afterwards, the measurement data have been generated by sampling from the quantum state vector or the MPS, respectively.
In one spatial dimension this can be done efficiently~\cite{han_unsupervised_2018,ferris_sampling_2012}.
The model we use for our analysis is a one-dimensional Heisenberg chain consisting of $n$ spins (and hence sites) 
\begin{equation}
  H = \sum_i \left(J^x X_iX_{i+1} + J^yY_iY_{i+1} + J^zZ_iZ_{i+1} + h_iX_i\right).
\end{equation}
We set the three interaction parameters to the arbitrary values $J^x, J^y, J^z = -1, -0.5, -0.4$, while we draw the local field parameters $h_i$  uniformly at random from the interval $[-1, 1]$.
As such, there are $\nu=3+n$ unknown parameters.
For each system size $n$, one parametrization is randomly drawn and used to simulate the time evolution.
At time stamps in intervals of $\tau=0.2$ (i.e., at $t=\tau, 2\tau, 3\tau, \ldots$) the state is then ``measured'' by sampling from the time-evolved state, typically in 100 uniformly randomly drawn Pauli bases.
\cref{fig:results} displays the performance of the method as a function of various experimental parameters.

\cref{fig:results}(a) shows a $d^{-1/2}$ error scaling which is in line with \cref{thm:error-informal}.
Note that the data points significantly fluctuate around the orange line.
The reason for this is that for each system size only \emph{one} Hamiltonian $H(\theta^*)$ was investigated.
\cref{fig:results}(b) indicates that at least for the Hamiltonian at hand, the error does not depend on the system size or the number of parameters $\nu$.
\cref{fig:results}(c) clearly shows a larger amount of outliers (optimization processes that did not converge to the target $\theta^*$) when only one Pauli basis is used to measure the system.
Meanwhile the accuracy is not affected by the number of Pauli bases.
In contrast, as shown in Fig.~\ref{fig:results}(d) a greater number of time stamps leads to greater accuracy, but also to a higher number of outliers.
As the number of time stamps grows, the loss landscape seems to become more rugged, while the minimum of the loss function moves closer to the true parameter values.
That is, the accuracy of the estimator given by the global minimum improves, but that minimum becomes more difficult to find.
See Appendix~\ref{sec:landscape} for additional discussions on the loss landscape.
Given these results and the fact that one cannot guarantee efficient optimization of the non-convex loss function, we conclude that the number of initial points has to be considered as a hyper parameter, which has to be adjusted according to the number of outliers---as witnessed by the loss function---one encounters during learning.
  
\paragraph*{Conclusion.}
In this work, we have introduced a scalable and experimentally friendly method of Hamiltonian learning based on the \emph{native time evolution} only.
Bringing together ideas and methods from machine learning and tensor networks, we have arrived at a highly practical scheme.
In a way, the methods developed are reminiscent of methods for learning classical dynamical laws by making use of tensor networks \cite{Klus,NonlinearDynamicalLaws}.
In contrast to other scalable methods based on dynamics for very small times~\cite{zubida_optimized_2021,yu_practical_2022}, our method achieves an error scaling of $d^{-1/2}$ as opposed to $d^{-1/3}$ or even $d^{-1/4}$ in the size of the data set.
Furthermore, our results suggest that our method does not require the preparation of multiple distinct initial states and even the number of Pauli bases in the measurement can be kept small.
Therefore, also in experiments where one has limited control over the system or extra manipulations (such as the preparation of specific states) are expensive and introduce additional noise, our method seems to be highly suitable.
It does not require further and possibly infeasible state manipulation such as Pauli twirls as other methods do, but uses the native time evolution.

We believe that our method is placed at a ``sweet spot'' concerning the duration of time evolution: 
For \emph{very short times}, expansions in powers of the Hamiltonian are feasible, so that the time evolution and its derivative can be estimated, but the signal-to-noise ratio is low, leading to a high sample complexity.
For very \emph{long times}, only certain families of quantum systems can be efficiently simulated \cite[e.g.,][]{GoogleHamiltonianLearningShort}, while in 
general the simulation of time evolution is computationally hard~\cite{Vollbrecht}.
In contrast, our method is based on intermediate times, leading to large signal strength that can at the same time be efficiently simulated using TEBD. 

The framework introduced here immediately invites a number of exciting extensions.
It would be interesting to extend the current framework to \emph{time-dependent Hamiltonians}.
Even more pressing seems the generalization to methods that are able to capture \emph{dissipative systems} where the target is to learn the \emph{Lindbladian noise} or the rate of dissipation~\cite{bairey_learning_2019}, derived from tensor network methods  simulating such dynamics~\cite{WeimerRMP,ZwolakPRL2004,KshetrimayumNatComm2017}.
Further work building on the present approach will aim at incorporating SPAM robustness, akin to the methods of Refs.~\cite{GoogleHamiltonianLearningShort,yu_practical_2022}.
It will be interesting to explore methods of model selection to systematically identify meaningful hypothesis classes to give advice when choosing the parametrization.
However, since the gradient-based optimization is efficient in the number of parameters, we can choose highly expressive parametrizations.

The picture that emerges from this work is that one can efficiently in space
and up to intermediate times
learn Hamiltonians of one-dimensional quantum many-body systems.
This approach is therefore well suited to analogue simulation platforms such as trapped ions~\cite{blatt_quantum_2012} or ultra-cold atoms~\cite{BlochSimulation}, in which either the state preparation and/or the measurement is flexible.
Based on this data, one can recover the Hamiltonian, 
but cannot efficiently make predictions for long times~\cite{Vollbrecht}: For this, one \emph{has} to perform the experiment to have predictive power, but---using our method---for an accurately calibrated Hamiltonian. 


\paragraph*{Acknowledgements.}
This work has been supported by the 
DFG (EI 519 20-1 on Hamiltonian learning,
MATH+ Cluster of Excellence
on notions of quantum machine learning
EF1-11, EI 519/15-1 on tensor 
networks and EXC-2046/1 – project ID: 390685689, 
GRK 2433 `Daedalus' on dynamical laws, as well as CRC 183, project B01, 
on notions of machine learning
applied to quantum many-body physics), the BMBF 
(DAQC on digital-analog quantum devices,
MUNIQC-ATOMS on cold atomic platforms),  
the BMWK (PlanQK on applied quantum machine learning), and the Einstein Foundation (Einstein Research Unit on Quantum Devices).
This research is also part of the Munich Quantum Valley (K8), which is supported by the Bavarian state government with funds from the Hightech Agenda Bayern Plus.
It has also received funding from 
the EU's Horizon 2020 research and innovation program
under grant agreement No.~817482 (PASQuanS) on quantum simulations. D.~H.\ acknowledges financial support from the U.S.\ Department of Defense through a QuICS Hartree fellowship. The authors would like to thank the HPC Service of ZEDAT, Freie Universität Berlin, for computing time \cite{hpc_zedat_fu}.

\bibliographystyle{myapsrev4-1}
\bibliography{system_identification}

\begin{thebibliography}{58}%
\makeatletter
\providecommand \@ifxundefined [1]{%
 \@ifx{#1\undefined}
}%
\providecommand \@ifnum [1]{%
 \ifnum #1\expandafter \@firstoftwo
 \else \expandafter \@secondoftwo
 \fi
}%
\providecommand \@ifx [1]{%
 \ifx #1\expandafter \@firstoftwo
 \else \expandafter \@secondoftwo
 \fi
}%
\providecommand \natexlab [1]{#1}%
\providecommand \enquote  [1]{``#1''}%
\providecommand \bibnamefont  [1]{#1}%
\providecommand \bibfnamefont [1]{#1}%
\providecommand \citenamefont [1]{#1}%
\providecommand \href@noop [0]{\@secondoftwo}%
\providecommand \href [0]{\begingroup \@sanitize@url \@href}%
\providecommand \@href[1]{\@@startlink{#1}\@@href}%
\providecommand \@@href[1]{\endgroup#1\@@endlink}%
\providecommand \@sanitize@url [0]{\catcode `\\12\catcode `\$12\catcode
  `\&12\catcode `\#12\catcode `\^12\catcode `\_12\catcode `\%12\relax}%
\providecommand \@@startlink[1]{}%
\providecommand \@@endlink[0]{}%
\providecommand \url  [0]{\begingroup\@sanitize@url \@url }%
\providecommand \@url [1]{\endgroup\@href {#1}{\urlprefix }}%
\providecommand \urlprefix  [0]{URL }%
\providecommand \Eprint [0]{\href }%
\providecommand \doibase [0]{http://dx.doi.org/}%
\providecommand \selectlanguage [0]{\@gobble}%
\providecommand \bibinfo  [0]{\@secondoftwo}%
\providecommand \bibfield  [0]{\@secondoftwo}%
\providecommand \translation [1]{[#1]}%
\providecommand \BibitemOpen [0]{}%
\providecommand \bibitemStop [0]{}%
\providecommand \bibitemNoStop [0]{.\EOS\space}%
\providecommand \EOS [0]{\spacefactor3000\relax}%
\providecommand \BibitemShut  [1]{\csname bibitem#1\endcsname}%
\let\auto@bib@innerbib\@empty
\bibitem [{\citenamefont {Cirac}\ and\ \citenamefont
  {Zoller}(2012)}]{CiracZollerSimulation}%
  \BibitemOpen
  \bibfield  {author} {\bibinfo {author} {\bibfnamefont {J.~I.}\ \bibnamefont
  {Cirac}}\ and\ \bibinfo {author} {\bibfnamefont {P.}~\bibnamefont {Zoller}},\
  }\bibinfo {title} {\emph {Goals and opportunities in quantum simulation}},\
  \href {\doibase 10.1038/nphys2275} {\bibfield  {journal} {\bibinfo  {journal}
  {Nature Phys.}\ }\textbf {\bibinfo {volume} {8}},\ \bibinfo {pages} {264}
  (\bibinfo {year} {2012})}\BibitemShut {NoStop}%
\bibitem [{\citenamefont {Trotzky}\ \emph {et~al.}(2012)\citenamefont
  {Trotzky}, \citenamefont {Chen}, \citenamefont {Flesch}, \citenamefont
  {McCulloch}, \citenamefont {Schollw\"ock}, \citenamefont {Eisert},\ and\
  \citenamefont {Bloch}}]{Trotzky}%
  \BibitemOpen
  \bibfield  {author} {\bibinfo {author} {\bibfnamefont {S.}~\bibnamefont
  {Trotzky}}, \bibinfo {author} {\bibfnamefont {Y.-A.}\ \bibnamefont {Chen}},
  \bibinfo {author} {\bibfnamefont {A.}~\bibnamefont {Flesch}}, \bibinfo
  {author} {\bibfnamefont {I.~P.}\ \bibnamefont {McCulloch}}, \bibinfo {author}
  {\bibfnamefont {U.}~\bibnamefont {Schollw\"ock}}, \bibinfo {author}
  {\bibfnamefont {J.}~\bibnamefont {Eisert}}, \ and\ \bibinfo {author}
  {\bibfnamefont {I.}~\bibnamefont {Bloch}},\ }\bibinfo {title} {\emph {Probing
  the relaxation towards equilibrium in an isolated strongly correlated
  one-dimensional {B}ose gas}},\ \href {\doibase doi:10.1038/nphys2232}
  {\bibfield  {journal} {\bibinfo  {journal} {Nature Phys.}\ }\textbf {\bibinfo
  {volume} {8}},\ \bibinfo {pages} {325} (\bibinfo {year} {2012})}\BibitemShut
  {NoStop}%
\bibitem [{\citenamefont {Choi}\ \emph {et~al.}(2016)\citenamefont {Choi},
  \citenamefont {Hild}, \citenamefont {Zeiher}, \citenamefont {Schau{\ss}},
  \citenamefont {Rubio-Abadal}, \citenamefont {Yefsah}, \citenamefont
  {Khemani}, \citenamefont {Huse}, \citenamefont {A.}, \citenamefont {Bloch},\
  and\ \citenamefont {Gross}}]{2dMBLScience}%
  \BibitemOpen
  \bibfield  {author} {\bibinfo {author} {\bibfnamefont {J.-Y.}\ \bibnamefont
  {Choi}}, \bibinfo {author} {\bibfnamefont {S.}~\bibnamefont {Hild}}, \bibinfo
  {author} {\bibfnamefont {J.}~\bibnamefont {Zeiher}}, \bibinfo {author}
  {\bibfnamefont {P.}~\bibnamefont {Schau{\ss}}}, \bibinfo {author}
  {\bibfnamefont {A.}~\bibnamefont {Rubio-Abadal}}, \bibinfo {author}
  {\bibfnamefont {T.}~\bibnamefont {Yefsah}}, \bibinfo {author} {\bibfnamefont
  {V.}~\bibnamefont {Khemani}}, \bibinfo {author} {\bibfnamefont
  {D.}~\bibnamefont {Huse}}, \bibinfo {author} {\bibnamefont {A.}}, \bibinfo
  {author} {\bibfnamefont {I.}~\bibnamefont {Bloch}}, \ and\ \bibinfo {author}
  {\bibfnamefont {C.}~\bibnamefont {Gross}},\ }\bibinfo {title} {\emph
  {Exploring the many-body localization transition in two dimensions}},\ \href
  {\doibase 10.1126/science.aaf883} {\bibfield  {journal} {\bibinfo  {journal}
  {Science}\ }\textbf {\bibinfo {volume} {352}},\ \bibinfo {pages} {1547}
  (\bibinfo {year} {2016})}\BibitemShut {NoStop}%
\bibitem [{\citenamefont {Carrasco}\ \emph {et~al.}(2021)\citenamefont
  {Carrasco}, \citenamefont {Elben}, \citenamefont {Kokail}, \citenamefont
  {Kraus},\ and\ \citenamefont {Zoller}}]{Carrasco_Verification_2021}%
  \BibitemOpen
  \bibfield  {author} {\bibinfo {author} {\bibfnamefont {J.}~\bibnamefont
  {Carrasco}}, \bibinfo {author} {\bibfnamefont {A.}~\bibnamefont {Elben}},
  \bibinfo {author} {\bibfnamefont {C.}~\bibnamefont {Kokail}}, \bibinfo
  {author} {\bibfnamefont {B.}~\bibnamefont {Kraus}}, \ and\ \bibinfo {author}
  {\bibfnamefont {P.}~\bibnamefont {Zoller}},\ }\bibinfo {title} {\emph
  {Theoretical and experimental perspectives of quantum verification}},\ \href
  {\doibase 10.1103/PRXQuantum.2.010102} {\bibfield  {journal} {\bibinfo
  {journal} {PRX Quantum}\ }\textbf {\bibinfo {volume} {2}},\ \bibinfo {pages}
  {010102} (\bibinfo {year} {2021})}\BibitemShut {NoStop}%
\bibitem [{\citenamefont {Krastanov}\ \emph {et~al.}(2019)\citenamefont
  {Krastanov}, \citenamefont {Zhou}, \citenamefont {Flammia},\ and\
  \citenamefont {Jiang}}]{krastanov_stochastic_2019}%
  \BibitemOpen
  \bibfield  {author} {\bibinfo {author} {\bibfnamefont {S.}~\bibnamefont
  {Krastanov}}, \bibinfo {author} {\bibfnamefont {S.}~\bibnamefont {Zhou}},
  \bibinfo {author} {\bibfnamefont {S.~T.}\ \bibnamefont {Flammia}}, \ and\
  \bibinfo {author} {\bibfnamefont {L.}~\bibnamefont {Jiang}},\ }\bibinfo
  {title} {\emph {Stochastic {estimation} of {dynamical} {variables}}},\ \href
  {\doibase 10.1088/2058-9565/ab18d5} {\bibfield  {journal} {\bibinfo
  {journal} {Quant Sc. Tech.}\ }\textbf {\bibinfo {volume} {4}},\ \bibinfo
  {pages} {035003} (\bibinfo {year} {2019})}\BibitemShut {NoStop}%
\bibitem [{\citenamefont {Liao}\ \emph {et~al.}(2019)\citenamefont {Liao},
  \citenamefont {Liu}, \citenamefont {Wang},\ and\ \citenamefont
  {Xiang}}]{liao_differentiable_2019}%
  \BibitemOpen
  \bibfield  {author} {\bibinfo {author} {\bibfnamefont {H.-J.}\ \bibnamefont
  {Liao}}, \bibinfo {author} {\bibfnamefont {J.-G.}\ \bibnamefont {Liu}},
  \bibinfo {author} {\bibfnamefont {L.}~\bibnamefont {Wang}}, \ and\ \bibinfo
  {author} {\bibfnamefont {T.}~\bibnamefont {Xiang}},\ }\bibinfo {title} {\emph
  {Differentiable {programming} {tensor} {networks}}},\ \href {\doibase
  10.1103/PhysRevX.9.031041} {\bibfield  {journal} {\bibinfo  {journal} {Phys.
  Rev. X}\ }\textbf {\bibinfo {volume} {9}},\ \bibinfo {pages} {031041}
  (\bibinfo {year} {2019})}\BibitemShut {NoStop}%
\bibitem [{Note1()}]{Note1}%
  \BibitemOpen
  \bibinfo {note} {See Ref.~\cite {Carrasco_Verification_2021} for a recent
  perspective article.}\BibitemShut {Stop}%
\bibitem [{\citenamefont {Bairey}\ \emph {et~al.}(2019)\citenamefont {Bairey},
  \citenamefont {Arad},\ and\ \citenamefont {Lindner}}]{bairey_learning_2019}%
  \BibitemOpen
  \bibfield  {author} {\bibinfo {author} {\bibfnamefont {E.}~\bibnamefont
  {Bairey}}, \bibinfo {author} {\bibfnamefont {I.}~\bibnamefont {Arad}}, \ and\
  \bibinfo {author} {\bibfnamefont {N.~H.}\ \bibnamefont {Lindner}},\ }\bibinfo
  {title} {\emph {Learning a {local} {Hamiltonian} from {local}
  {measurements}}},\ \href {\doibase 10.1103/PhysRevLett.122.020504} {\bibfield
   {journal} {\bibinfo  {journal} {Phys. Rev. Lett.}\ }\textbf {\bibinfo
  {volume} {122}},\ \bibinfo {pages} {020504} (\bibinfo {year}
  {2019})}\BibitemShut {NoStop}%
\bibitem [{\citenamefont {{Bairey}}\ \emph {et~al.}(2020)\citenamefont
  {{Bairey}}, \citenamefont {{Guo}}, \citenamefont {{Poletti}}, \citenamefont
  {{Lindner}},\ and\ \citenamefont {{Arad}}}]{BaireyNJP}%
  \BibitemOpen
  \bibfield  {author} {\bibinfo {author} {\bibfnamefont {E.}~\bibnamefont
  {{Bairey}}}, \bibinfo {author} {\bibfnamefont {C.}~\bibnamefont {{Guo}}},
  \bibinfo {author} {\bibfnamefont {D.}~\bibnamefont {{Poletti}}}, \bibinfo
  {author} {\bibfnamefont {N.~H.}\ \bibnamefont {{Lindner}}}, \ and\ \bibinfo
  {author} {\bibfnamefont {I.}~\bibnamefont {{Arad}}},\ }\bibinfo {title}
  {\emph {{Learning the dynamics of open quantum systems from their steady
  states}}},\ \href {\doibase 10.1088/1367-2630/ab73cd} {\bibfield  {journal}
  {\bibinfo  {journal} {New J. Phys.}\ }\textbf {\bibinfo {volume} {22}},\
  \bibinfo {eid} {032001} (\bibinfo {year} {2020})}\BibitemShut {NoStop}%
\bibitem [{\citenamefont {Anshu}\ \emph {et~al.}(2021)\citenamefont {Anshu},
  \citenamefont {Arunachalam}, \citenamefont {Kuwahara},\ and\ \citenamefont
  {Soleimanifar}}]{anshu_sample-efficient_2021}%
  \BibitemOpen
  \bibfield  {author} {\bibinfo {author} {\bibfnamefont {A.}~\bibnamefont
  {Anshu}}, \bibinfo {author} {\bibfnamefont {S.}~\bibnamefont {Arunachalam}},
  \bibinfo {author} {\bibfnamefont {T.}~\bibnamefont {Kuwahara}}, \ and\
  \bibinfo {author} {\bibfnamefont {M.}~\bibnamefont {Soleimanifar}},\
  }\bibinfo {title} {\emph {Sample-efficient learning of interacting quantum
  systems}},\ \href {\doibase 10.1038/s41567-021-01232-0} {\bibfield  {journal}
  {\bibinfo  {journal} {Nat. Phys.}\ ,\ \bibinfo {pages} {1}} (\bibinfo {year}
  {2021})}\BibitemShut {NoStop}%
\bibitem [{\citenamefont {Haah}\ \emph {et~al.}(2021)\citenamefont {Haah},
  \citenamefont {Kothari},\ and\ \citenamefont {Tang}}]{haah_optimal_2021}%
  \BibitemOpen
  \bibfield  {author} {\bibinfo {author} {\bibfnamefont {J.}~\bibnamefont
  {Haah}}, \bibinfo {author} {\bibfnamefont {R.}~\bibnamefont {Kothari}}, \
  and\ \bibinfo {author} {\bibfnamefont {E.}~\bibnamefont {Tang}},\ }\bibinfo
  {title} {\emph {Optimal learning of quantum {{Hamiltonians}} from
  high-temperature {{Gibbs}} states}},\ \href@noop {} {\  (\bibinfo {year}
  {2021})},\ \Eprint {http://arxiv.org/abs/2108.04842} {arXiv:2108.04842
  [quant-ph]}\BibitemShut {NoStop}%
\bibitem [{\citenamefont {Qi}\ and\ \citenamefont
  {Ranard}(2019)}]{qi_determining_2019}%
  \BibitemOpen
  \bibfield  {author} {\bibinfo {author} {\bibfnamefont {X.-L.}\ \bibnamefont
  {Qi}}\ and\ \bibinfo {author} {\bibfnamefont {D.}~\bibnamefont {Ranard}},\
  }\bibinfo {title} {\emph {Determining a Local {{Hamiltonian}} from a Single
  Eigenstate}},\ \href {\doibase 10.22331/q-2019-07-08-159} {\bibfield
  {journal} {\bibinfo  {journal} {Quantum}\ }\textbf {\bibinfo {volume} {3}},\
  \bibinfo {pages} {159} (\bibinfo {year} {2019})}\BibitemShut {NoStop}%
\bibitem [{\citenamefont {Zubida}\ \emph {et~al.}(2021)\citenamefont {Zubida},
  \citenamefont {Yitzhaki}, \citenamefont {Lindner},\ and\ \citenamefont
  {Bairey}}]{zubida_optimized_2021}%
  \BibitemOpen
  \bibfield  {author} {\bibinfo {author} {\bibfnamefont {A.}~\bibnamefont
  {Zubida}}, \bibinfo {author} {\bibfnamefont {E.}~\bibnamefont {Yitzhaki}},
  \bibinfo {author} {\bibfnamefont {N.~H.}\ \bibnamefont {Lindner}}, \ and\
  \bibinfo {author} {\bibfnamefont {E.}~\bibnamefont {Bairey}},\ }\bibinfo
  {title} {\emph {Optimized {Hamiltonian} learning from short-time
  measurements}},\ \href@noop {} {\  (\bibinfo {year} {2021})},\ \Eprint
  {http://arxiv.org/abs/2108.08824} {arXiv:2108.08824}\BibitemShut {NoStop}%
\bibitem [{\citenamefont {França}\ \emph {et~al.}(2022)\citenamefont
  {França}, \citenamefont {Markovich}, \citenamefont {Dobrovitski},
  \citenamefont {Werner},\ and\ \citenamefont
  {Borregaard}}]{franca_efficient_2022}%
  \BibitemOpen
  \bibfield  {author} {\bibinfo {author} {\bibfnamefont {D.~S.}\ \bibnamefont
  {França}}, \bibinfo {author} {\bibfnamefont {L.~A.}\ \bibnamefont
  {Markovich}}, \bibinfo {author} {\bibfnamefont {V.~V.}\ \bibnamefont
  {Dobrovitski}}, \bibinfo {author} {\bibfnamefont {A.~H.}\ \bibnamefont
  {Werner}}, \ and\ \bibinfo {author} {\bibfnamefont {J.}~\bibnamefont
  {Borregaard}},\ }\bibinfo {title} {\emph {Efficient and robust estimation of
  many-qubit {Hamiltonians}}},\ \href@noop {} {\  (\bibinfo {year} {2022})},\
  \Eprint {http://arxiv.org/abs/2205.09567} {arXiv:2205.09567}\BibitemShut
  {NoStop}%
\bibitem [{\citenamefont {Gu}\ \emph {et~al.}(2022)\citenamefont {Gu},
  \citenamefont {Cincio},\ and\ \citenamefont {Coles}}]{gu_practical_2022}%
  \BibitemOpen
  \bibfield  {author} {\bibinfo {author} {\bibfnamefont {A.}~\bibnamefont
  {Gu}}, \bibinfo {author} {\bibfnamefont {L.}~\bibnamefont {Cincio}}, \ and\
  \bibinfo {author} {\bibfnamefont {P.~J.}\ \bibnamefont {Coles}},\ }\bibinfo
  {title} {\emph {Practical {black} {box} {Hamiltonian} {learning}}},\
  \href@noop {} {\  (\bibinfo {year} {2022})},\ \Eprint
  {http://arxiv.org/abs/2206.15464} {arXiv:2206.15464}\BibitemShut {NoStop}%
\bibitem [{\citenamefont {Harper}\ \emph {et~al.}(2021)\citenamefont {Harper},
  \citenamefont {Yu},\ and\ \citenamefont {Flammia}}]{harper_fast_2021}%
  \BibitemOpen
  \bibfield  {author} {\bibinfo {author} {\bibfnamefont {R.}~\bibnamefont
  {Harper}}, \bibinfo {author} {\bibfnamefont {W.}~\bibnamefont {Yu}}, \ and\
  \bibinfo {author} {\bibfnamefont {S.~T.}\ \bibnamefont {Flammia}},\ }\bibinfo
  {title} {\emph {Fast {{estimation}} of {{sparse quantum noise}}}},\ \href
  {\doibase 10.1103/PRXQuantum.2.010322} {\bibfield  {journal} {\bibinfo
  {journal} {PRX Quantum}\ }\textbf {\bibinfo {volume} {2}},\ \bibinfo {pages}
  {010322} (\bibinfo {year} {2021})}\BibitemShut {NoStop}%
\bibitem [{\citenamefont {Yu}\ \emph {et~al.}(2022)\citenamefont {Yu},
  \citenamefont {Sun}, \citenamefont {Han},\ and\ \citenamefont
  {Yuan}}]{yu_practical_2022}%
  \BibitemOpen
  \bibfield  {author} {\bibinfo {author} {\bibfnamefont {W.}~\bibnamefont
  {Yu}}, \bibinfo {author} {\bibfnamefont {J.}~\bibnamefont {Sun}}, \bibinfo
  {author} {\bibfnamefont {Z.}~\bibnamefont {Han}}, \ and\ \bibinfo {author}
  {\bibfnamefont {X.}~\bibnamefont {Yuan}},\ }\bibinfo {title} {\emph
  {Practical and {efficient} {Hamiltonian} {learning}}},\ \href@noop {} {\
  (\bibinfo {year} {2022})},\ \Eprint {http://arxiv.org/abs/2201.00190}
  {arXiv:2201.00190}\BibitemShut {NoStop}%
\bibitem [{\citenamefont {Oi}\ and\ \citenamefont
  {Schirmer}(2012)}]{oi_quantum_2012}%
  \BibitemOpen
  \bibfield  {author} {\bibinfo {author} {\bibfnamefont {D.~K.~L.}\
  \bibnamefont {Oi}}\ and\ \bibinfo {author} {\bibfnamefont {S.~G.}\
  \bibnamefont {Schirmer}},\ }\bibinfo {title} {\emph {Quantum system
  characterization with limited resources}},\ \href {\doibase
  10.1098/rsta.2011.0530} {\bibfield  {journal} {\bibinfo  {journal} {Phil.
  Trans. R. Soc. A}\ }\textbf {\bibinfo {volume} {370}},\ \bibinfo {pages}
  {5386} (\bibinfo {year} {2012})}\BibitemShut {NoStop}%
\bibitem [{\citenamefont {Burgarth}\ \emph {et~al.}(2011)\citenamefont
  {Burgarth}, \citenamefont {Maruyama},\ and\ \citenamefont
  {Nori}}]{burgarth_indirect_2011}%
  \BibitemOpen
  \bibfield  {author} {\bibinfo {author} {\bibfnamefont {D.}~\bibnamefont
  {Burgarth}}, \bibinfo {author} {\bibfnamefont {K.}~\bibnamefont {Maruyama}},
  \ and\ \bibinfo {author} {\bibfnamefont {F.}~\bibnamefont {Nori}},\ }\bibinfo
  {title} {\emph {Indirect quantum tomography of quadratic {{Hamiltonians}}}},\
  \href {\doibase 10.1088/1367-2630/13/1/013019} {\bibfield  {journal}
  {\bibinfo  {journal} {New J. Phys.}\ }\textbf {\bibinfo {volume} {13}},\
  \bibinfo {pages} {013019} (\bibinfo {year} {2011})}\BibitemShut {NoStop}%
\bibitem [{\citenamefont {Hangleiter}\ \emph {et~al.}(2021)\citenamefont
  {Hangleiter}, \citenamefont {Roth}, \citenamefont {Eisert},\ and\
  \citenamefont {Roushan}}]{GoogleHamiltonianLearningShort}%
  \BibitemOpen
  \bibfield  {author} {\bibinfo {author} {\bibfnamefont {D.}~\bibnamefont
  {Hangleiter}}, \bibinfo {author} {\bibfnamefont {I.}~\bibnamefont {Roth}},
  \bibinfo {author} {\bibfnamefont {J.}~\bibnamefont {Eisert}}, \ and\ \bibinfo
  {author} {\bibfnamefont {P.}~\bibnamefont {Roushan}},\ }\bibinfo {title}
  {\emph {Precise Hamiltonian identification of a superconducting quantum
  processor}},\ \href@noop {} {\  (\bibinfo {year} {2021})},\ \Eprint
  {http://arxiv.org/abs/2108.08319} {arXiv:2108.08319}\BibitemShut {NoStop}%
\bibitem [{\citenamefont {Zhang}\ and\ \citenamefont
  {Sarovar}(2014)}]{zhang_quantum_2014}%
  \BibitemOpen
  \bibfield  {author} {\bibinfo {author} {\bibfnamefont {J.}~\bibnamefont
  {Zhang}}\ and\ \bibinfo {author} {\bibfnamefont {M.}~\bibnamefont
  {Sarovar}},\ }\bibinfo {title} {\emph {Quantum {{Hamiltonian identification}}
  from {{measurement time traces}}}},\ \href {\doibase
  10.1103/PhysRevLett.113.080401} {\bibfield  {journal} {\bibinfo  {journal}
  {Phys. Rev. Lett.}\ }\textbf {\bibinfo {volume} {113}},\ \bibinfo {pages}
  {080401} (\bibinfo {year} {2014})}\BibitemShut {NoStop}%
\bibitem [{\citenamefont {Chen}\ \emph {et~al.}(2021)\citenamefont {Chen},
  \citenamefont {Li}, \citenamefont {Wu}, \citenamefont {Liu}, \citenamefont
  {Li},\ and\ \citenamefont {Zhou}}]{chen_experimental_2021}%
  \BibitemOpen
  \bibfield  {author} {\bibinfo {author} {\bibfnamefont {X.}~\bibnamefont
  {Chen}}, \bibinfo {author} {\bibfnamefont {Y.}~\bibnamefont {Li}}, \bibinfo
  {author} {\bibfnamefont {Z.}~\bibnamefont {Wu}}, \bibinfo {author}
  {\bibfnamefont {R.}~\bibnamefont {Liu}}, \bibinfo {author} {\bibfnamefont
  {Z.}~\bibnamefont {Li}}, \ and\ \bibinfo {author} {\bibfnamefont
  {H.}~\bibnamefont {Zhou}},\ }\bibinfo {title} {\emph {Experimental
  realization of {{Hamiltonian}} tomography by quantum quenches}},\ \href
  {\doibase 10.1103/PhysRevA.103.042429} {\bibfield  {journal} {\bibinfo
  {journal} {Phys. Rev. A}\ }\textbf {\bibinfo {volume} {103}},\ \bibinfo
  {pages} {042429} (\bibinfo {year} {2021})}\BibitemShut {NoStop}%
\bibitem [{\citenamefont {Li}\ \emph {et~al.}(2020)\citenamefont {Li},
  \citenamefont {Zou},\ and\ \citenamefont {Hsieh}}]{li_hamiltonian_2020}%
  \BibitemOpen
  \bibfield  {author} {\bibinfo {author} {\bibfnamefont {Z.}~\bibnamefont
  {Li}}, \bibinfo {author} {\bibfnamefont {L.}~\bibnamefont {Zou}}, \ and\
  \bibinfo {author} {\bibfnamefont {T.~H.}\ \bibnamefont {Hsieh}},\ }\bibinfo
  {title} {\emph {Hamiltonian {tomography} via {quantum} {quench}}},\ \href
  {\doibase 10.1103/PhysRevLett.124.160502} {\bibfield  {journal} {\bibinfo
  {journal} {Phys. Rev. Lett.}\ }\textbf {\bibinfo {volume} {124}},\ \bibinfo
  {pages} {160502} (\bibinfo {year} {2020})}\BibitemShut {NoStop}%
\bibitem [{\citenamefont {Pastori}\ \emph {et~al.}(2022)\citenamefont
  {Pastori}, \citenamefont {Olsacher}, \citenamefont {Kokail},\ and\
  \citenamefont {Zoller}}]{pastori_characterization_2022}%
  \BibitemOpen
  \bibfield  {author} {\bibinfo {author} {\bibfnamefont {L.}~\bibnamefont
  {Pastori}}, \bibinfo {author} {\bibfnamefont {T.}~\bibnamefont {Olsacher}},
  \bibinfo {author} {\bibfnamefont {C.}~\bibnamefont {Kokail}}, \ and\ \bibinfo
  {author} {\bibfnamefont {P.}~\bibnamefont {Zoller}},\ }\bibinfo {title}
  {\emph {Characterization and verification of Trotterized digital quantum
  simulation via Hamiltonian and Liouvillian learning}},\ \href {\doibase
  10.1103/PRXQuantum.3.030324} {\bibfield  {journal} {\bibinfo  {journal} {PRX
  Quantum}\ }\textbf {\bibinfo {volume} {3}},\ \bibinfo {pages} {030324}
  (\bibinfo {year} {2022})}\BibitemShut {NoStop}%
\bibitem [{\citenamefont {Bienias}\ \emph {et~al.}(2021)\citenamefont
  {Bienias}, \citenamefont {Seif},\ and\ \citenamefont
  {Hafezi}}]{bienias_meta_2021}%
  \BibitemOpen
  \bibfield  {author} {\bibinfo {author} {\bibfnamefont {P.}~\bibnamefont
  {Bienias}}, \bibinfo {author} {\bibfnamefont {A.}~\bibnamefont {Seif}}, \
  and\ \bibinfo {author} {\bibfnamefont {M.}~\bibnamefont {Hafezi}},\ }\bibinfo
  {title} {\emph {Meta {{Hamiltonian learning}}}},\ \href@noop {} {\  (\bibinfo
  {year} {2021})},\ \Eprint {http://arxiv.org/abs/2104.04453}
  {arXiv:2104.04453}\BibitemShut {NoStop}%
\bibitem [{\citenamefont {Schuster}\ \emph {et~al.}(2022)\citenamefont
  {Schuster}, \citenamefont {Niu}, \citenamefont {Cotler}, \citenamefont
  {O'Brien}, \citenamefont {McClean},\ and\ \citenamefont
  {Mohseni}}]{schuster_learning_2022}%
  \BibitemOpen
  \bibfield  {author} {\bibinfo {author} {\bibfnamefont {T.}~\bibnamefont
  {Schuster}}, \bibinfo {author} {\bibfnamefont {M.}~\bibnamefont {Niu}},
  \bibinfo {author} {\bibfnamefont {J.}~\bibnamefont {Cotler}}, \bibinfo
  {author} {\bibfnamefont {T.}~\bibnamefont {O'Brien}}, \bibinfo {author}
  {\bibfnamefont {J.~R.}\ \bibnamefont {McClean}}, \ and\ \bibinfo {author}
  {\bibfnamefont {M.}~\bibnamefont {Mohseni}},\ }\bibinfo {title} {\emph
  {Learning quantum systems via out-of-time-order correlators}},\ \href@noop {}
  {\  (\bibinfo {year} {2022})},\ \Eprint {http://arxiv.org/abs/2208.02254}
  {arXiv:2208.02254}\BibitemShut {NoStop}%
\bibitem [{\citenamefont {Mohseni}\ \emph {et~al.}(2022)\citenamefont
  {Mohseni}, \citenamefont {F{\"{o}}sel}, \citenamefont {Guo}, \citenamefont
  {Navarrete-Benlloch},\ and\ \citenamefont {Marquardt}}]{mohseni_deep_2021}%
  \BibitemOpen
  \bibfield  {author} {\bibinfo {author} {\bibfnamefont {N.}~\bibnamefont
  {Mohseni}}, \bibinfo {author} {\bibfnamefont {T.}~\bibnamefont
  {F{\"{o}}sel}}, \bibinfo {author} {\bibfnamefont {L.}~\bibnamefont {Guo}},
  \bibinfo {author} {\bibfnamefont {C.}~\bibnamefont {Navarrete-Benlloch}}, \
  and\ \bibinfo {author} {\bibfnamefont {F.}~\bibnamefont {Marquardt}},\
  }\bibinfo {title} {\emph {Deep {l}earning of {q}uantum {m}any-{b}ody
  {d}ynamics via {r}andom {d}riving}},\ \href {\doibase
  10.22331/q-2022-05-17-714} {\bibfield  {journal} {\bibinfo  {journal}
  {{Quantum}}\ }\textbf {\bibinfo {volume} {6}},\ \bibinfo {pages} {714}
  (\bibinfo {year} {2022})}\BibitemShut {NoStop}%
\bibitem [{\citenamefont {Han}\ \emph {et~al.}(2021)\citenamefont {Han},
  \citenamefont {Glaz}, \citenamefont {Haile},\ and\ \citenamefont
  {Lai}}]{han_tomography_2021}%
  \BibitemOpen
  \bibfield  {author} {\bibinfo {author} {\bibfnamefont {C.-D.}\ \bibnamefont
  {Han}}, \bibinfo {author} {\bibfnamefont {B.}~\bibnamefont {Glaz}}, \bibinfo
  {author} {\bibfnamefont {M.}~\bibnamefont {Haile}}, \ and\ \bibinfo {author}
  {\bibfnamefont {Y.-C.}\ \bibnamefont {Lai}},\ }\bibinfo {title} {\emph
  {Tomography of time-dependent quantum {Hamiltonians} with machine
  learning}},\ \href {\doibase 10.1103/PhysRevA.104.062404} {\bibfield
  {journal} {\bibinfo  {journal} {Phys. Rev. A}\ }\textbf {\bibinfo {volume}
  {104}},\ \bibinfo {pages} {062404} (\bibinfo {year} {2021})}\BibitemShut
  {NoStop}%
\bibitem [{\citenamefont {Valenti}\ \emph {et~al.}(2022)\citenamefont
  {Valenti}, \citenamefont {Jin}, \citenamefont {Léonard}, \citenamefont
  {Huber},\ and\ \citenamefont {Greplova}}]{valenti_scalable_2022}%
  \BibitemOpen
  \bibfield  {author} {\bibinfo {author} {\bibfnamefont {A.}~\bibnamefont
  {Valenti}}, \bibinfo {author} {\bibfnamefont {G.}~\bibnamefont {Jin}},
  \bibinfo {author} {\bibfnamefont {J.}~\bibnamefont {Léonard}}, \bibinfo
  {author} {\bibfnamefont {S.~D.}\ \bibnamefont {Huber}}, \ and\ \bibinfo
  {author} {\bibfnamefont {E.}~\bibnamefont {Greplova}},\ }\bibinfo {title}
  {\emph {Scalable {Hamiltonian} learning for large-scale out-of-equilibrium
  quantum dynamics}},\ \href {\doibase 10.1103/PhysRevA.105.023302} {\bibfield
  {journal} {\bibinfo  {journal} {Phys. Rev. A}\ }\textbf {\bibinfo {volume}
  {105}},\ \bibinfo {pages} {023302} (\bibinfo {year} {2022})}\BibitemShut
  {NoStop}%
\bibitem [{\citenamefont {Xie}\ \emph {et~al.}(2022)\citenamefont {Xie},
  \citenamefont {Dai}, \citenamefont {Yuan}, \citenamefont {Deng},
  \citenamefont {Li}, \citenamefont {Chen},\ and\ \citenamefont
  {Pan}}]{xie_bayesian_2022}%
  \BibitemOpen
  \bibfield  {author} {\bibinfo {author} {\bibfnamefont {Y.-J.}\ \bibnamefont
  {Xie}}, \bibinfo {author} {\bibfnamefont {H.-N.}\ \bibnamefont {Dai}},
  \bibinfo {author} {\bibfnamefont {Z.-S.}\ \bibnamefont {Yuan}}, \bibinfo
  {author} {\bibfnamefont {Y.}~\bibnamefont {Deng}}, \bibinfo {author}
  {\bibfnamefont {X.}~\bibnamefont {Li}}, \bibinfo {author} {\bibfnamefont
  {Y.-A.}\ \bibnamefont {Chen}}, \ and\ \bibinfo {author} {\bibfnamefont
  {J.-W.}\ \bibnamefont {Pan}},\ }\bibinfo {title} {\emph {Bayesian learning
  for optimal control of quantum many-body states in optical lattices}},\ \href
  {\doibase 10.1103/PhysRevA.106.013316} {\bibfield  {journal} {\bibinfo
  {journal} {Phys. Rev. A}\ }\textbf {\bibinfo {volume} {106}},\ \bibinfo
  {pages} {013316} (\bibinfo {year} {2022})}\BibitemShut {NoStop}%
\bibitem [{\citenamefont {Evans}\ \emph {et~al.}(2019)\citenamefont {Evans},
  \citenamefont {Harper},\ and\ \citenamefont {Flammia}}]{evans_scalable_2019}%
  \BibitemOpen
  \bibfield  {author} {\bibinfo {author} {\bibfnamefont {T.~J.}\ \bibnamefont
  {Evans}}, \bibinfo {author} {\bibfnamefont {R.}~\bibnamefont {Harper}}, \
  and\ \bibinfo {author} {\bibfnamefont {S.~T.}\ \bibnamefont {Flammia}},\
  }\bibinfo {title} {\emph {Scalable {Bayesian} {Hamiltonian} learning}},\
  \href@noop {} {\  (\bibinfo {year} {2019})},\ \Eprint
  {http://arxiv.org/abs/1912.07636} {arXiv:1912.07636}\BibitemShut {NoStop}%
\bibitem [{\citenamefont {Schollwöck}(2011)}]{schollwock_density-matrix_2011}%
  \BibitemOpen
  \bibfield  {author} {\bibinfo {author} {\bibfnamefont {U.}~\bibnamefont
  {Schollwöck}},\ }\bibinfo {title} {\emph {The density-matrix renormalization
  group in the age of matrix product states}},\ \href {\doibase
  10.1016/j.aop.2010.09.012} {\bibfield  {journal} {\bibinfo  {journal} {Ann.
  Phys.}\ }\bibinfo {series} {January 2011 {Special} {Issue}},\ \textbf
  {\bibinfo {volume} {326}},\ \bibinfo {pages} {96} (\bibinfo {year}
  {2011})}\BibitemShut {NoStop}%
\bibitem [{\citenamefont {Vidal}(2004)}]{Vidaltebd}%
  \BibitemOpen
  \bibfield  {author} {\bibinfo {author} {\bibfnamefont {G.}~\bibnamefont
  {Vidal}},\ }\bibinfo {title} {\emph {Efficient simulation of one-dimensional
  quantum many-body systems}},\ \href {\doibase 10.1103/PhysRevLett.93.040502}
  {\bibfield  {journal} {\bibinfo  {journal} {Phys. Rev. Lett.}\ }\textbf
  {\bibinfo {volume} {93}},\ \bibinfo {pages} {040502} (\bibinfo {year}
  {2004})}\BibitemShut {NoStop}%
\bibitem [{\citenamefont {Zwolak}\ and\ \citenamefont
  {Vidal}(2004)}]{ZwolakPRL2004}%
  \BibitemOpen
  \bibfield  {author} {\bibinfo {author} {\bibfnamefont {M.}~\bibnamefont
  {Zwolak}}\ and\ \bibinfo {author} {\bibfnamefont {G.}~\bibnamefont {Vidal}},\
  }\bibinfo {title} {\emph {Mixed-state dynamics in one-dimensional quantum
  lattice systems: A time-dependent superoperator renormalization algorithm}},\
  \href {\doibase 10.1103/PhysRevLett.93.207205} {\bibfield  {journal}
  {\bibinfo  {journal} {Phys. Rev. Lett.}\ }\textbf {\bibinfo {volume} {93}},\
  \bibinfo {pages} {207205} (\bibinfo {year} {2004})}\BibitemShut {NoStop}%
\bibitem [{\citenamefont {Kshetrimayum}\ \emph {et~al.}(2020)\citenamefont
  {Kshetrimayum}, \citenamefont {Goihl},\ and\ \citenamefont
  {Eisert}}]{Kshetrimayum2dMBL2020}%
  \BibitemOpen
  \bibfield  {author} {\bibinfo {author} {\bibfnamefont {A.}~\bibnamefont
  {Kshetrimayum}}, \bibinfo {author} {\bibfnamefont {M.}~\bibnamefont {Goihl}},
  \ and\ \bibinfo {author} {\bibfnamefont {J.}~\bibnamefont {Eisert}},\
  }\bibinfo {title} {\emph {Time evolution of many-body localized systems in
  two spatial dimensions}},\ \href {\doibase 10.1103/PhysRevB.102.235132}
  {\bibfield  {journal} {\bibinfo  {journal} {Phys. Rev. B}\ }\textbf {\bibinfo
  {volume} {102}},\ \bibinfo {pages} {235132} (\bibinfo {year}
  {2020})}\BibitemShut {NoStop}%
\bibitem [{\citenamefont {Hubig}\ and\ \citenamefont
  {Cirac}(2019)}]{Hubigpepsdynamics}%
  \BibitemOpen
  \bibfield  {author} {\bibinfo {author} {\bibfnamefont {C.}~\bibnamefont
  {Hubig}}\ and\ \bibinfo {author} {\bibfnamefont {J.~I.}\ \bibnamefont
  {Cirac}},\ }\bibinfo {title} {\emph {{Time-dependent study of disordered
  models with infinite projected entangled pair states}}},\ \href {\doibase
  10.21468/SciPostPhys.6.3.031} {\bibfield  {journal} {\bibinfo  {journal}
  {SciPost Phys.}\ }\textbf {\bibinfo {volume} {6}},\ \bibinfo {pages} {031}
  (\bibinfo {year} {2019})}\BibitemShut {NoStop}%
\bibitem [{\citenamefont {Kshetrimayum}\ \emph {et~al.}(2021)\citenamefont
  {Kshetrimayum}, \citenamefont {Goihl}, \citenamefont {Kennes},\ and\
  \citenamefont {Eisert}}]{Kshetrimayum2dTC2021}%
  \BibitemOpen
  \bibfield  {author} {\bibinfo {author} {\bibfnamefont {A.}~\bibnamefont
  {Kshetrimayum}}, \bibinfo {author} {\bibfnamefont {M.}~\bibnamefont {Goihl}},
  \bibinfo {author} {\bibfnamefont {D.~M.}\ \bibnamefont {Kennes}}, \ and\
  \bibinfo {author} {\bibfnamefont {J.}~\bibnamefont {Eisert}},\ }\bibinfo
  {title} {\emph {Quantum time crystals with programmable disorder in higher
  dimensions}},\ \href {\doibase 10.1103/PhysRevB.103.224205} {\bibfield
  {journal} {\bibinfo  {journal} {Phys. Rev. B}\ }\textbf {\bibinfo {volume}
  {103}},\ \bibinfo {pages} {224205} (\bibinfo {year} {2021})}\BibitemShut
  {NoStop}%
\bibitem [{\citenamefont {Bradbury}\ \emph {et~al.}(2018)\citenamefont
  {Bradbury}, \citenamefont {Frostig}, \citenamefont {Hawkins}, \citenamefont
  {Johnson}, \citenamefont {Leary}, \citenamefont {Maclaurin}, \citenamefont
  {Necula}, \citenamefont {Paszke}, \citenamefont {Vander{P}las}, \citenamefont
  {Wanderman-{M}ilne},\ and\ \citenamefont {Zhang}}]{jax2018github}%
  \BibitemOpen
  \bibfield  {author} {\bibinfo {author} {\bibfnamefont {J.}~\bibnamefont
  {Bradbury}}, \bibinfo {author} {\bibfnamefont {R.}~\bibnamefont {Frostig}},
  \bibinfo {author} {\bibfnamefont {P.}~\bibnamefont {Hawkins}}, \bibinfo
  {author} {\bibfnamefont {M.~J.}\ \bibnamefont {Johnson}}, \bibinfo {author}
  {\bibfnamefont {C.}~\bibnamefont {Leary}}, \bibinfo {author} {\bibfnamefont
  {D.}~\bibnamefont {Maclaurin}}, \bibinfo {author} {\bibfnamefont
  {G.}~\bibnamefont {Necula}}, \bibinfo {author} {\bibfnamefont
  {A.}~\bibnamefont {Paszke}}, \bibinfo {author} {\bibfnamefont
  {J.}~\bibnamefont {Vander{P}las}}, \bibinfo {author} {\bibfnamefont
  {S.}~\bibnamefont {Wanderman-{M}ilne}}, \ and\ \bibinfo {author}
  {\bibfnamefont {Q.}~\bibnamefont {Zhang}},\ }\href
  {http://github.com/google/jax} {\bibinfo {title} {\emph {{JAX}: composable
  transformations of {P}ython+{N}um{P}y programs}},\ } (\bibinfo {year}
  {2018})\BibitemShut {NoStop}%
\bibitem [{git()}]{github_dynamical-hamiltonian-learning_nodate}%
  \BibitemOpen
  \href
  {https://github.com/frederikwilde/scalable-dynamical-hamiltonian-learning}
  {\bibinfo {title} {\emph
  {github.com/frederikwilde/scalable-dynamical-hamiltonian-learning}},\
  }\BibitemShut {NoStop}%
\bibitem [{Note2()}]{Note2}%
  \BibitemOpen
  \bibinfo {note} {We use the Runge-Kutta ODE solver implemented in \protect
  \texttt {scipy.integrate.solve\protect \_ivp} to integrate the
  Schr{\"o}dinger equation in its exact form on the Hilbert space of dimension
  $2^n$.}\BibitemShut {Stop}%
\bibitem [{\citenamefont {Han}\ \emph {et~al.}(2018)\citenamefont {Han},
  \citenamefont {Wang}, \citenamefont {Fan}, \citenamefont {Wang},\ and\
  \citenamefont {Zhang}}]{han_unsupervised_2018}%
  \BibitemOpen
  \bibfield  {author} {\bibinfo {author} {\bibfnamefont {Z.-Y.}\ \bibnamefont
  {Han}}, \bibinfo {author} {\bibfnamefont {J.}~\bibnamefont {Wang}}, \bibinfo
  {author} {\bibfnamefont {H.}~\bibnamefont {Fan}}, \bibinfo {author}
  {\bibfnamefont {L.}~\bibnamefont {Wang}}, \ and\ \bibinfo {author}
  {\bibfnamefont {P.}~\bibnamefont {Zhang}},\ }\bibinfo {title} {\emph
  {Unsupervised {generative} {modeling} {using} {matrix} {product} {states}}},\
  \href {\doibase 10.1103/PhysRevX.8.031012} {\bibfield  {journal} {\bibinfo
  {journal} {Phys. Rev. X}\ }\textbf {\bibinfo {volume} {8}},\ \bibinfo {pages}
  {031012} (\bibinfo {year} {2018})}\BibitemShut {NoStop}%
\bibitem [{\citenamefont {Ferris}\ and\ \citenamefont
  {Vidal}(2012)}]{ferris_sampling_2012}%
  \BibitemOpen
  \bibfield  {author} {\bibinfo {author} {\bibfnamefont {A.~J.}\ \bibnamefont
  {Ferris}}\ and\ \bibinfo {author} {\bibfnamefont {G.}~\bibnamefont {Vidal}},\
  }\bibinfo {title} {\emph {Perfect sampling with unitary tensor networks}},\
  \href {\doibase 10.1103/PhysRevB.85.165146} {\bibfield  {journal} {\bibinfo
  {journal} {Phys. Rev. B}\ }\textbf {\bibinfo {volume} {85}},\ \bibinfo
  {pages} {165146} (\bibinfo {year} {2012})}\BibitemShut {NoStop}%
\bibitem [{\citenamefont {Gel\ss}\ \emph {et~al.}(2019)\citenamefont {Gel\ss},
  \citenamefont {Klus}, \citenamefont {Eisert},\ and\ \citenamefont
  {Sch{\"u}tte}}]{Klus}%
  \BibitemOpen
  \bibfield  {author} {\bibinfo {author} {\bibfnamefont {P.}~\bibnamefont
  {Gel\ss}}, \bibinfo {author} {\bibfnamefont {S.}~\bibnamefont {Klus}},
  \bibinfo {author} {\bibfnamefont {J.}~\bibnamefont {Eisert}}, \ and\ \bibinfo
  {author} {\bibfnamefont {C.}~\bibnamefont {Sch{\"u}tte}},\ }\bibinfo {title}
  {\emph {Multidimensional approximation of nonlinear dynamical systems}},\
  \href {\doibase 10.1115/1.4043148} {\bibfield  {journal} {\bibinfo  {journal}
  {J. Comput. Nonlinear Dynam.}\ }\textbf {\bibinfo {volume} {14}},\ \bibinfo
  {pages} {061006} (\bibinfo {year} {2019})}\BibitemShut {NoStop}%
\bibitem [{\citenamefont {Goe{\ss}mann}\ \emph {et~al.}(2020)\citenamefont
  {Goe{\ss}mann}, \citenamefont {G{\"o}tte}, \citenamefont {Roth},
  \citenamefont {Sweke}, \citenamefont {Kutyniok},\ and\ \citenamefont
  {Eisert}}]{NonlinearDynamicalLaws}%
  \BibitemOpen
  \bibfield  {author} {\bibinfo {author} {\bibfnamefont {A.}~\bibnamefont
  {Goe{\ss}mann}}, \bibinfo {author} {\bibfnamefont {M.}~\bibnamefont
  {G{\"o}tte}}, \bibinfo {author} {\bibfnamefont {I.}~\bibnamefont {Roth}},
  \bibinfo {author} {\bibfnamefont {R.}~\bibnamefont {Sweke}}, \bibinfo
  {author} {\bibfnamefont {G.}~\bibnamefont {Kutyniok}}, \ and\ \bibinfo
  {author} {\bibfnamefont {J.}~\bibnamefont {Eisert}},\ }\bibinfo {title}
  {\emph {Tensor network approaches for learning non-linear dynamical laws}},\
  \href@noop {} {\  (\bibinfo {year} {2020})},\ \Eprint
  {http://arxiv.org/abs/2002.12388} {arXiv:2002.12388}\BibitemShut {NoStop}%
\bibitem [{\citenamefont {Vollbrecht}\ and\ \citenamefont
  {Cirac}(2008)}]{Vollbrecht}%
  \BibitemOpen
  \bibfield  {author} {\bibinfo {author} {\bibfnamefont {K.~G.~H.}\
  \bibnamefont {Vollbrecht}}\ and\ \bibinfo {author} {\bibfnamefont {J.~I.}\
  \bibnamefont {Cirac}},\ }\bibinfo {title} {\emph {Quantum simulators,
  continuous-time automata, and translationally invariant systems}},\ \href
  {\doibase 10.1103/PhysRevLett.100.010501} {\bibfield  {journal} {\bibinfo
  {journal} {Phys. Rev. Lett.}\ }\textbf {\bibinfo {volume} {100}},\ \bibinfo
  {pages} {010501} (\bibinfo {year} {2008})}\BibitemShut {NoStop}%
\bibitem [{\citenamefont {Weimer}\ \emph {et~al.}(2021)\citenamefont {Weimer},
  \citenamefont {Kshetrimayum},\ and\ \citenamefont {Or\'us}}]{WeimerRMP}%
  \BibitemOpen
  \bibfield  {author} {\bibinfo {author} {\bibfnamefont {H.}~\bibnamefont
  {Weimer}}, \bibinfo {author} {\bibfnamefont {A.}~\bibnamefont
  {Kshetrimayum}}, \ and\ \bibinfo {author} {\bibfnamefont {R.}~\bibnamefont
  {Or\'us}},\ }\bibinfo {title} {\emph {Simulation methods for open quantum
  many-body systems}},\ \href {\doibase 10.1103/RevModPhys.93.015008}
  {\bibfield  {journal} {\bibinfo  {journal} {Rev. Mod. Phys.}\ }\textbf
  {\bibinfo {volume} {93}},\ \bibinfo {pages} {015008} (\bibinfo {year}
  {2021})}\BibitemShut {NoStop}%
\bibitem [{\citenamefont {Kshetrimayum}\ \emph {et~al.}(2017)\citenamefont
  {Kshetrimayum}, \citenamefont {Weimer},\ and\ \citenamefont
  {Or{\'u}s}}]{KshetrimayumNatComm2017}%
  \BibitemOpen
  \bibfield  {author} {\bibinfo {author} {\bibfnamefont {A.}~\bibnamefont
  {Kshetrimayum}}, \bibinfo {author} {\bibfnamefont {H.}~\bibnamefont
  {Weimer}}, \ and\ \bibinfo {author} {\bibfnamefont {R.}~\bibnamefont
  {Or{\'u}s}},\ }\bibinfo {title} {\emph {A simple tensor network algorithm for
  two-dimensional steady states}},\ \href {\doibase 10.1038/s41467-017-01511-6}
  {\bibfield  {journal} {\bibinfo  {journal} {Nature Comm.}\ }\textbf {\bibinfo
  {volume} {8}},\ \bibinfo {pages} {1291} (\bibinfo {year} {2017})}\BibitemShut
  {NoStop}%
\bibitem [{\citenamefont {Blatt}\ and\ \citenamefont
  {Roos}(2012)}]{blatt_quantum_2012}%
  \BibitemOpen
  \bibfield  {author} {\bibinfo {author} {\bibfnamefont {R.}~\bibnamefont
  {Blatt}}\ and\ \bibinfo {author} {\bibfnamefont {C.~F.}\ \bibnamefont
  {Roos}},\ }\bibinfo {title} {\emph {Quantum simulations with trapped ions}},\
  \href {\doibase 10.1038/nphys2252} {\bibfield  {journal} {\bibinfo  {journal}
  {Nature Phys.}\ }\textbf {\bibinfo {volume} {8}},\ \bibinfo {pages} {277}
  (\bibinfo {year} {2012})}\BibitemShut {NoStop}%
\bibitem [{\citenamefont {Bloch}\ \emph {et~al.}(2012)\citenamefont {Bloch},
  \citenamefont {Dalibard},\ and\ \citenamefont
  {Nascimbene}}]{BlochSimulation}%
  \BibitemOpen
  \bibfield  {author} {\bibinfo {author} {\bibfnamefont {I.}~\bibnamefont
  {Bloch}}, \bibinfo {author} {\bibfnamefont {J.}~\bibnamefont {Dalibard}}, \
  and\ \bibinfo {author} {\bibfnamefont {S.}~\bibnamefont {Nascimbene}},\
  }\bibinfo {title} {\emph {Quantum simulations with ultracold quantum
  gases}},\ \href {\doibase 10.1038/nphys2259} {\bibfield  {journal} {\bibinfo
  {journal} {Nature Phys.}\ }\textbf {\bibinfo {volume} {8}},\ \bibinfo {pages}
  {267} (\bibinfo {year} {2012})}\BibitemShut {NoStop}%
\bibitem [{\citenamefont {Bennett}\ \emph {et~al.}(2020)\citenamefont
  {Bennett}, \citenamefont {Melchers},\ and\ \citenamefont
  {Proppe}}]{hpc_zedat_fu}%
  \BibitemOpen
  \bibfield  {author} {\bibinfo {author} {\bibfnamefont {L.}~\bibnamefont
  {Bennett}}, \bibinfo {author} {\bibfnamefont {B.}~\bibnamefont {Melchers}}, \
  and\ \bibinfo {author} {\bibfnamefont {B.}~\bibnamefont {Proppe}},\
  }\href@noop {} {\bibinfo {title} {\emph {Curta: A General-purpose
  high-performance computer at {ZEDAT}, {F}reie {U}niversität {B}erlin}},\ }
  (\bibinfo {year} {2020})\BibitemShut {NoStop}%
\bibitem [{\citenamefont {Bellman}(1987)}]{bellmann_matrix_analysis}%
  \BibitemOpen
  \bibfield  {author} {\bibinfo {author} {\bibfnamefont {R.}~\bibnamefont
  {Bellman}},\ }\href@noop {} {\emph {\bibinfo {title} {Introduction to matrix
  analysis}}}\ (\bibinfo  {publisher} {Society for Industrial \& Applied
  Mathematics},\ \bibinfo {year} {1987})\BibitemShut {NoStop}%
\bibitem [{\citenamefont {Newey}\ and\ \citenamefont
  {McFadden}(2005)}]{newey_large_nodate}%
  \BibitemOpen
  \bibfield  {author} {\bibinfo {author} {\bibfnamefont {W.~K.}\ \bibnamefont
  {Newey}}\ and\ \bibinfo {author} {\bibfnamefont {D.}~\bibnamefont
  {McFadden}},\ }\bibinfo {title} {\emph {Large sample estimation and
  hypothesis}},\ \href {\doibase 10.1016/S1573-4412(05)80005-4} {\bibfield
  {journal} {\bibinfo  {journal} {Handbook Econom.}\ }\textbf {\bibinfo
  {volume} {36}},\ \bibinfo {pages} {135} (\bibinfo {year} {2005})}\BibitemShut
  {NoStop}%
\bibitem [{\citenamefont {Wan}\ and\ \citenamefont
  {Zhang}(2019)}]{wan_automatic_2019}%
  \BibitemOpen
  \bibfield  {author} {\bibinfo {author} {\bibfnamefont {Z.-Q.}\ \bibnamefont
  {Wan}}\ and\ \bibinfo {author} {\bibfnamefont {S.-X.}\ \bibnamefont
  {Zhang}},\ }\bibinfo {title} {\emph {Automatic {differentiation} for
  {complex} {valued} {SVD}}},\ \href@noop {} {\  (\bibinfo {year} {2019})},\
  \Eprint {http://arxiv.org/abs/1909.02659} {arXiv:1909.02659}\BibitemShut
  {NoStop}%
\bibitem [{\citenamefont {Liu}(2019)}]{liu_linear_2019}%
  \BibitemOpen
  \bibfield  {author} {\bibinfo {author} {\bibfnamefont {J.-G.}\ \bibnamefont
  {Liu}},\ }\href {https://giggleliu.github.io/2019/04/02/einsumbp.html}
  {\bibinfo {title} {\emph {Linear {algebra} {autodiff} (complex valued)}},\ }
  (\bibinfo {year} {2019})\BibitemShut {NoStop}%
\bibitem [{\citenamefont {Giles}(2007)}]{giles_extended_nodate}%
  \BibitemOpen
  \bibfield  {author} {\bibinfo {author} {\bibfnamefont {M.}~\bibnamefont
  {Giles}},\ }\href@noop {} {\bibinfo {title} {\emph {An extended collection of
  matrix derivative results for forward and reverse mode algorithmic
  differentiation}},\ } (\bibinfo {year} {2007}),\ \bibinfo {note} {tech. Rep.
  NA07}\BibitemShut {NoStop}%
\bibitem [{\citenamefont {Bartholomew-Biggs}\ \emph {et~al.}(2000)\citenamefont
  {Bartholomew-Biggs}, \citenamefont {Brown}, \citenamefont {Christianson},\
  and\ \citenamefont {Dixon}}]{bartholomew-biggs_automatic_2000}%
  \BibitemOpen
  \bibfield  {author} {\bibinfo {author} {\bibfnamefont {M.}~\bibnamefont
  {Bartholomew-Biggs}}, \bibinfo {author} {\bibfnamefont {S.}~\bibnamefont
  {Brown}}, \bibinfo {author} {\bibfnamefont {B.}~\bibnamefont {Christianson}},
  \ and\ \bibinfo {author} {\bibfnamefont {L.}~\bibnamefont {Dixon}},\
  }\bibinfo {title} {\emph {Automatic differentiation of algorithms}},\ \href
  {\doibase 10.1016/S0377-0427(00)00422-2} {\bibfield  {journal} {\bibinfo
  {journal} {J. Comp. Appl. Math.}\ }\textbf {\bibinfo {volume} {124}},\
  \bibinfo {pages} {171} (\bibinfo {year} {2000})}\BibitemShut {NoStop}%
\bibitem [{tot()}]{total_derivative_footnote}%
  \BibitemOpen
  \href@noop {} {}\bibinfo {note} {We can imagine $A$ being dependent on some
  parameter $p$ such that we can identify $\dd A$, $\dd U$, $\dd S$, and
  $\dd\dg{V}$ with $\partial A / \partial p$, $\partial U / \partial p$,
  $\partial S / \partial p$, and $\partial \dg{V} / \partial p$,
  respectively.}\BibitemShut {Stop}%
\bibitem [{\citenamefont {Townsend}(2016)}]{townsend_differentiating_2016}%
  \BibitemOpen
  \bibfield  {author} {\bibinfo {author} {\bibfnamefont {J.}~\bibnamefont
  {Townsend}},\ }\href {https://j-towns.github.io/papers/svd-derivative.pdf}
  {\bibinfo {title} {\emph {{D}ifferentiating the {singular} {value}
  {decomposition}}},\ } (\bibinfo {year} {2016})\BibitemShut {NoStop}%
\end{thebibliography}%


\appendix

\section{Details of the numerical computations}
\label{sec:implementation}
The key component of the loss function is the computation of probabilities of bit-strings.
This is achieved by creating an initial MPS in the $\ket{0}^{\otimes n}$ state vector and applying Trotter steps to it in an iterative fashion until the state has evolved up to the desired time.
We found a suitable value of the \emph{bond dimension}, used for all results presented here, to be $\chi = 30$.
For longer evolution times or a Hamiltonian that generates more entanglement, the bond dimension will have to be increased.
We implemented the \emph{second-order Trotter} formula, as this reduces the Trotter error with little computational overhead.
The time-evolved MPS can then be contracted with the post-measurement state, as depicted in Fig.~\ref{fig:setting}(b), to compute the probability of the corresponding bit-string.

To enable the efficient use of \emph{automatic differentiation} of the probabilities, it is important to keep the dimensions of the arrays fixed, during the execution of the TEBD algorithm.
Therefore the MPS is a fixed array of dimensions $(n, \chi, 2, \chi)$, where the physical dimension is 2 since we are considering a spin-1/2 system.
Furthermore, the algorithm provides several opportunities for \emph{vectorization}: we vectorize the probability computation over Pauli bases and bit-strings, the basis transformation of the MPS, and the application of individual Trotter gates.
Lastly, we \emph{just-in-time compile} the entire TEBD function.
This takes several seconds to complete in the first step (see Fig.~\ref{fig:timing}), but reduces the overall computation time significantly.
To limit excessive memory usage during back-propagation we \emph{checkpoint} each Trotter step.
In the case of mini-batch gradient descent checkpointing is not necessary, which reduces the time for evaluating the gradient.
However, for gradients on the full data set, without checkpointing the method quickly exceeds the memory constraints, as can be seen in Fig.~\ref{fig:timing}.
Even with checkpointing we had to iterate over chunks of the data set to compute the gradient, for larger data set sizes $d$.
For instance, for a system of $n=10$ spins consisting of $10^6$ samples per time step, we have divided the data set into 10 chunks which required about 22GiB of memory.
Therefore, we predict that for even larger data sets, system sizes, evolution times, and bond dimensions, due to memory requirements, it will become necessary to only optimize with stochastic gradients, perhaps with a tailored learning rate schedule.

Initially we investigated conjugate gradient descent, plain stochastic gradient descent, BFGS, and ADAM.
However the latter two appeared to be more reliable and faster in our context.
The results, shown in Fig.~\ref{fig:optimizers}, indicate that ADAM with a small batch size can quickly and reliably reduce the error, while BFGS can reduce the error reliably to the minimal value, however at longer computation times and with the necessity to evaluate the gradient w.r.t.\ the full data set.

Since the negative log-likelihood is a non-convex function, it is difficult to make guarantees on the number of iterations needed to find the global minimum.
However, we observe that the number of iterations shows only a weak dependence on the system size $n$, as shown in the left plot of Fig.~\ref{fig:timing}.
Therefore, the relevant figure of merit for determining the time complexity is the computation time needed for one iteration, i.e.\ one gradient evaluation.
As shown in the center plot of Fig.~\ref{fig:timing}, with sufficient memory this computation time scales linearly in system size.
Lastly, it is noteworthy that the system considered in our analysis the system size is connected to the number of parameters $\nu = 3 + n$.

\begin{figure}
    \includegraphics[width=.48\textwidth]{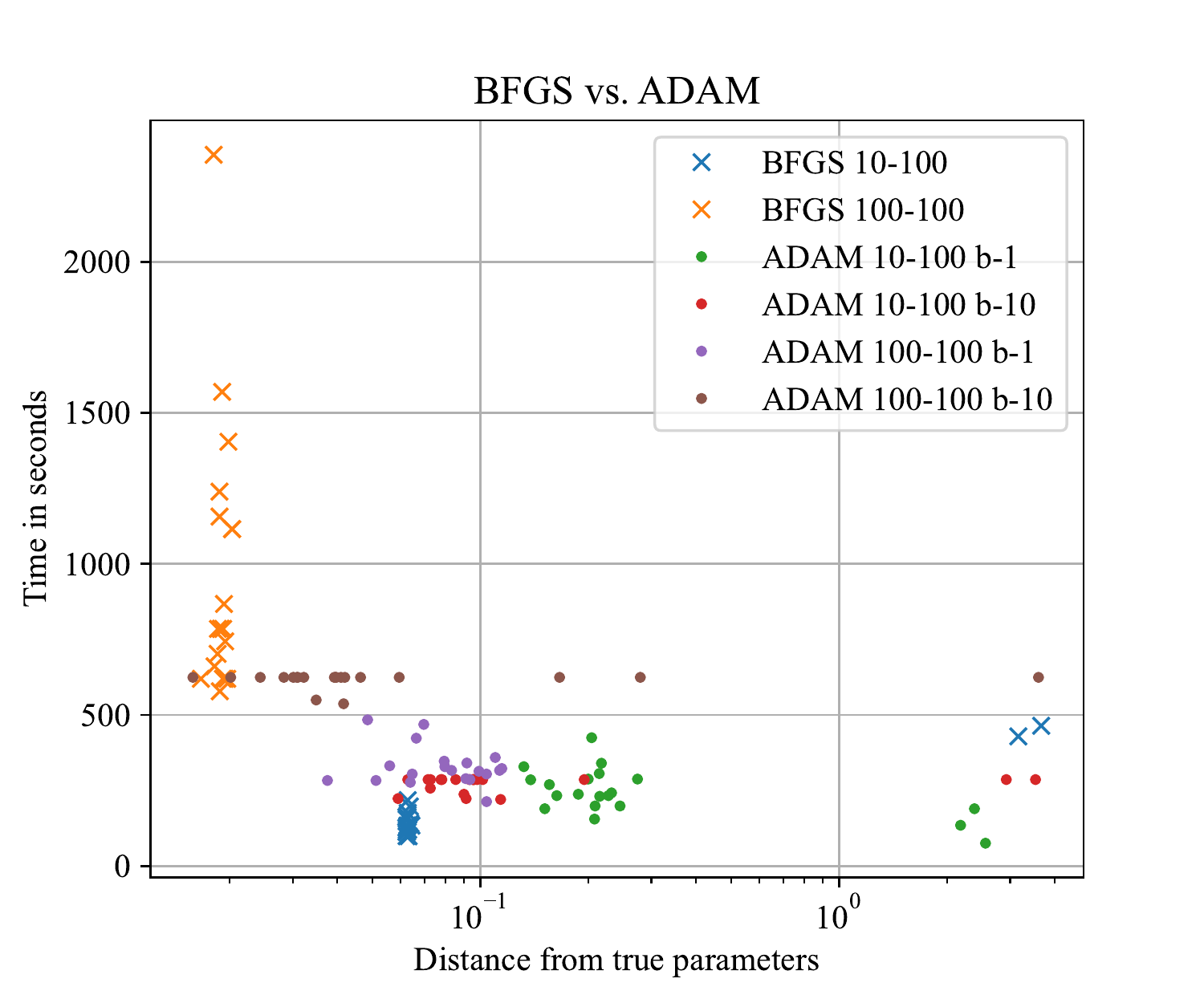}
    \caption{\label{fig:optimizers}
    Optimization times for different optimizers. The numbers $K$-$M$ denote the number of Pauli bases and bit-strings, respectively. b-$B$ denotes batch size $B$, which refers to the number of bit-strings used per Pauli basis in a given gradient computation.}
\end{figure}

\begin{figure*}
  \centering
  \includegraphics[width=\textwidth]{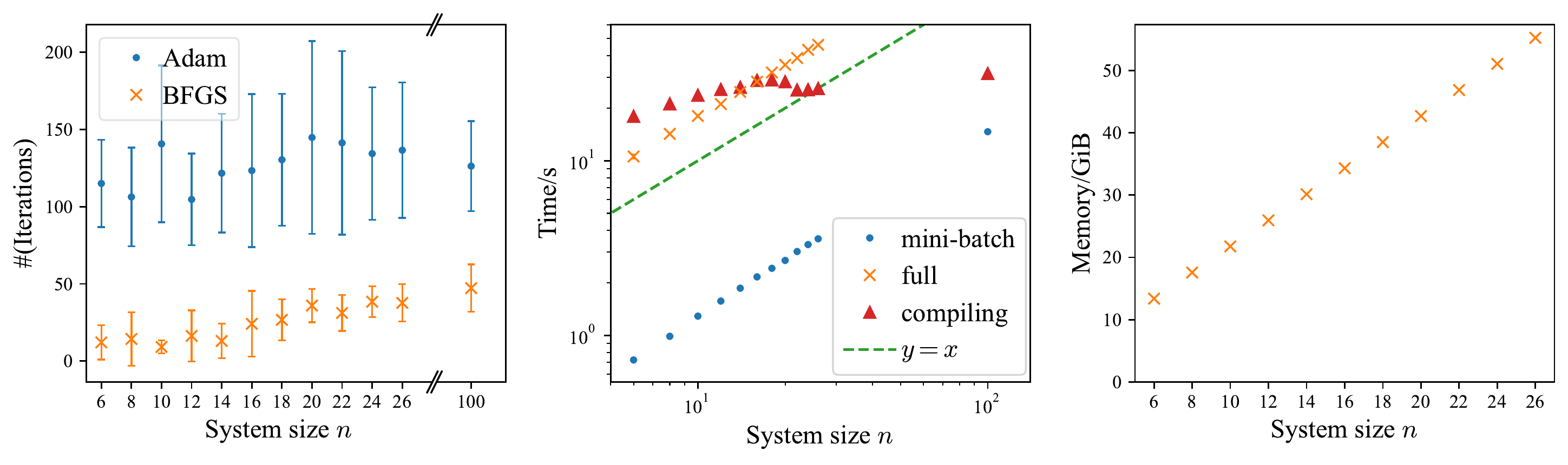}
  \caption{
    \label{fig:timing}
    Resource requirements for learning the Heisenberg Hamiltonian at varying system sizes $n$.
    \textbf{Left:} The number of iterations needed to converge.
    Note that our optimization algorithm consists of two stages: an initial phase of stochastic gradient descend using the ADAM optimizer followed by a fine-tuning phase to converge to the minimum using the BFGS optimizer.
    \textbf{Center:} Times for computing one gradient evaluation on a data set of size $d=50000$ bit-strings ($J, K, M =  5, 100, 100$).
    For the mini-batch stochastic gradient descent with ADAM only one bit-string per basis is randomly chosen.
    In this case we have an effective data-set size $d_\mathrm{eff}=500$.
    For $n=100$ sites the memory requirement for the gradient w.r.t.\ the full data set is prohibitively large.
    In practice we computed the gradient w.r.t.\ chunks of the data set and summed up the results subsequently.
    The computation time was recorded on an Intel Core i9-10900 CPU@2.80GHz processor.
    However, the program utilizes at most 4 cores simultaneously.
    The time required for JIT compiling the computation is only required once for a given system size.
    The function $y=x$ (green dashed line) has the same slope as the run times, which illustrates the linear scaling.
    \textbf{Right:} As the system size grows, so does the memory requirement of the gradient computation.
    The data shown is the requirement for computing the gradient w.r.t.\ the full data set.
    For mini-batch stochastic gradients the memory requirement is reduced by a factor 100.}
\end{figure*}

\section{Asymptotic error}
\label{sec:normality}
For this section, 
we use the abbreviated notation $\partial_l = {\partial}/{\partial \theta_l}$ and $\nabla = \nabla_\theta$.
Let us first precisely define the data set $\mathcal{D}$, its associated loss function $\loss^\mathcal{D}$, and the minimizer of the loss $\theta^{\mathcal{D}}$ function as random variables.
Suppose we are given non-empty sets of time stamps $\lbrace t_1, \ldots, t_J\rbrace$ and Pauli bases $\lbrace p_1, \ldots, p_K\rbrace \subseteq \lbrace X, Y, Z\rbrace^n$, where $t_j>0$ for all $j=1, \ldots, J$.
Moreover, we denote the family of target 
Hamiltonians---under which the data set is being generated---by $H(\theta^*)$.
Changing the target Hamiltonian, the Pauli bases, or the time stamps would obviously impact the measurement outcomes and therefore the obtained data set $\mathcal{D}$.
However, for the sake of lucidity, we 
do not explicity write out this dependency and assume that those three are fixed.

\begin{definition}[Data set as a random variable]
  Each Pauli basis $p_k$ gives rise to a positive projector-valued measure (PVM) $\lbrace \pi_{0,0,\dots, 0}^k, \ldots, \pi_{1,1,\dots, 1}^k \rbrace$.
  The PVM elements together with a time stamp $t_j$ give rise to a probability distribution (over bit-strings of length $n$)
  \begin{align*}
    \rho_{j,k}: \lbrace 0, 1\rbrace^n &\longrightarrow [0, 1] \\
    s &\longmapsto \mathrm{tr}(\pi_s^k\ket{\psi_j}\bra{\psi_j}),
  \end{align*}
  where $\ket{\psi_j} = \ee^{-\ii H(\theta^*)t_j}\ket{0}$ is the quantum state vector at time $t_j$.
  Suppose further that we select a set of positive sample numbers $\lbrace d_{j,k}\vert 1\leq j \leq J,~1\leq k\leq K\rbrace$, specifying how often we measure each Pauli basis at each time.
  Denote by $d = \sum_{j,k} d_{j,k}$ the total size of the data set.
  Now we define a realization $\mathcal{D}$ of the data-set random variable $\hat{\mathcal{D}}$ to be specified by a collection of samples (bit-strings of length $n$)
  \begin{align*}
    \mathcal{D} = \lbrace s_{i,j,k} \in \lbrace 0, 1\rbrace^n\vert 1\leq i\leq d_{j,k}, 1\leq j \leq J,~1\leq k\leq K\rbrace
  \end{align*}
  and the probability of observing that realization is
  \begin{align*}
    \mathbb{P}(\mathcal{D}) = \prod_{j,k}\prod_{i=1}^{d_{j,k}} \rho_{j,k}(s_{i,j,k}).
  \end{align*}
\end{definition}

The definition of the empirical loss $\hat{\mathcal{L}}^{(d)}$ and the minimizer $\hat{\theta}^{(d)}$ follow straightforwardly.
A realization of $\hat{\mathcal{L}}^{(d)}$ and $\hat{\theta}^{(d)}$ is given by $\mathcal{L}^\mathcal{D}$ and its minimizer $\theta^\mathcal{D}$, respectively, where $\mathcal{D}$ is a data set realization of size $d$.

\begin{definition}[Expected loss function]
  The expected negative log-likelihood in the limit $d\rightarrow\infty$ is formally given by
  \begin{align}
    \loss(\theta)
    := - \frac{1}{JK}\sum_{j,k}\sum_{s\in\lbrace 0, 1\rbrace^2}^\star \rho_{j,k}(s) \log(\rho_{j,k}(s\vert\theta)),
  \end{align}
  where $\star$ indicates that only terms satisfying $\rho_{j,k}(s)\neq 0$ appear in the sum.
  Here $\rho_{j,k}(s\vert \theta)$ denotes the probability of measuring $s$ in the basis $p_k$ under the parameter values $\theta$ at time $t_j$.
\end{definition}

Now we show that $\loss^\mathcal{D}$ is twice continuously differentiable for all data sets $\mathcal{D}$.
\begin{lemma}[Derivatives of the loss function]
  \label{lem:derivatives}
  The loss function $\loss^\mathcal{D}$ is twice differentiable in every component.
\end{lemma}
\begin{proof}
  We start with the observation that the directional derivative of the matrix exponential can be expressed as
  \begin{align}
    \label{eq:matrix-exp-derivative}
    \frac{\dd}{\dd h}\ee^{A+hV} = \int_0^1 \ee^{(1-\tau)A}V\ee^{\tau A} \dd \tau.
  \end{align}
  This can be verified as follows 
  \cite{bellmann_matrix_analysis}.
  $\ee^{A+hV}$ is the solution of the matrix-valued differential equation
  \begin{align}
    \frac{\dd X}{\dd t} = (A + hV) X, \quad X(0) = \mathds{1}
  \end{align}
  at time $t=1$.
  By subtracting $AX$ and multiplying by $\ee^{-At}$ from the left we can write the differential equation as
  \begin{align}
    \frac{\dd}{\dd t}(\ee^{-At}X) = \ee^{-At}hVX,
  \end{align}
  which has the formal solution
  \begin{align}
    \ee^{-At}X(t) = \mathds{1} + \int_0^t \ee^{-A\tau}hVX(\tau)\dd \tau.
  \end{align}
  Multiplying by $\ee^{At}$ from the left and 
  expanding the series we identify the first derivative of $X$ with respect $h$
  \begin{align}
    X(t) = \ee^{At} + h\underbrace{\int_0^t \ee^{A(t-\tau)}V\ee^{A\tau} \dd\tau}_{=\frac{\dd X}{\dd h}} + \mathcal{O}(h^2),
  \end{align}
  which, evaluated at $t=1$ proves the validity of Eq.~(\ref{eq:matrix-exp-derivative}).
  Now we note that in the loss
  \begin{align}
    \loss^\mathcal{D}(\theta) = - \frac{1}{d} \sum_{i,j,k} \log \rho_{j,k}(s_i \vert \theta),
  \end{align}
  the probabilities can be written as
  \begin{align}
    \rho_{j,k}(s_i \vert \theta) &:= \bra{\phi_{i,j,k}}\ee^{-\ii H(\theta)t}\ket{0} \bra{0}\ee^{\ii H(\theta)t}\ket{\phi_{i,j,k}},
  \end{align}
  where, in line with the main text, $\ket{\phi_{i,j,k}}$ is the post measurement state vector corresponding to $s_{i,j,k}$.
  In this form, 
  it becomes apparent that if $H$ is twice continuously differentiable, the same holds for the loss function.
\end{proof}

\begin{lemma}[Zero-mean property of the score]
  \label{lem:zero-mean}
  In the limit of large $d$, where each $d_{j,k}\rightarrow\infty$, 
  the scaled derivative of the loss converges in distribution to the normal distribution
  \begin{align}
    \sqrt{d}\, \nabla_\theta \hat{\loss}^{(d)}(\theta^*) \overset{\mathrm{dist.}}{\longrightarrow} \mathcal{N}(0, \Sigma).
  \end{align}
\end{lemma}
\begin{proof}
  Let $\hat{A}_{j,k} = \nabla_\theta \log\big(\rho_{j,k}(\hat{s}\vert \theta)\big)\big\vert_{\theta = \theta^*}$ denote the so called score, so that we can write
  \begin{align}
    \nabla_\theta \hat{\loss}^{(d)}(\theta^*) = \frac{1}{JK}\sum_{j,k} \hat{A}^{(d_{j,k})}_{j,k},
  \end{align}
  where $\hat{A}^{(N)}_{j,k}$ is the $N$-sample mean estimator of $\hat{A}_{j,k}$.
  To apply the central limit theorem to the score, we show that it has zero mean and finite variance.
  The mean is readily calculated to be
  \begin{align}
    \mathbb{E}(\hat{A}_{j,k}) &= \sum_{s\in\lbrace 0, 1\rbrace^n} \rho_{j,k}(s)\nabla_\theta\log\big(\rho_{j,k}(s\vert\theta)\big)\big\vert_{\theta = \theta^*} \\
    &= \nabla_\theta \underbrace{\sum_{s\in\lbrace 0, 1\rbrace^n} \rho_{j,k}(s)\vert\theta}_{=1}\Bigg\vert_{\theta = \theta^*} \\
    &= 0. \nonumber
  \end{align}
  The $(l,m)$-th element of the covariance matrix is given by
  \begin{align}
    &\mathrm{Cov}_{l,m}(\hat{A}_{j,k}) \\
    &= \sum_{s\in\lbrace 0, 1\rbrace^n} \rho^+_{j,k}(s) \big[\partial_l \rho_{j,k}(s\vert\theta)\big] \big[\partial_m \rho_{j,k}(s\vert\theta)\big]\bigg\vert_{\theta=\theta^*}, \nonumber
  \end{align}
  where $x^+$ denotes the pseudo inverse of $x$.
  Since $\rho_{j,k}$ is twice differentiable all summands are finite and thus the sum is finite.
  Now, by the multivariate central limit theorem $\hat{A}_{j,k}^{(d_{j,k})} \overset{\mathrm{dist.}}{\longrightarrow} \mathcal{N}\big(0, \mathrm{Cov}(\hat{A}_{j,k})\big)$ as $d_{j,k} \rightarrow \infty$.
  Therefore, the sum of all scores converges in distribution to $\mathcal{N}(0, \Sigma)$, where
  \begin{align}
    \Sigma = \frac{1}{JK}\sum_{j,k} \mathrm{Cov}(\hat{A}_{j,k}).
  \end{align}
\end{proof}

\begin{thm}[Normality of the maximum likelihood estimator \cite{newey_large_nodate}]
  \label{thm:normality}
  Suppose the following conditions hold:
  \begin{enumerate}
    \item The estimator $\hat{\theta}^{(d)}$ is consistent, i.e., $\hat{\theta}^{(d)} \overset{\mathrm{p.}}{\longrightarrow} \theta^*$.
    \item $\theta^* \in \mathrm{interior}(\Theta)$.
    \item The loss $\hat{\loss}^{(d)}$ is twice continuously differentiable in a neighborhood $N$ of $\theta^*$.
    \item The loss satisfies $\sqrt{d}\nabla_\theta \hat{\loss}^{(d)}(\theta^*) \overset{\mathrm{dist.}}{\longrightarrow} \mathcal{N}(0, \Sigma)$.
    \item There exists a function $H:\rr^p \rightarrow \rr^{p\times p}$ that is continuous at $\theta^*$ with the property
    $$\underset{\theta\in\mathcal{N}}{\mathrm{sup}} \Vert \nabla^2_\theta \hat{\loss}^{(d)}(\theta) - H(\theta) \Vert \overset{\mathrm{p.}}{\longrightarrow} 0.$$
    \item $H(\theta^*)$ is non-singular.
  \end{enumerate}
  Then the estimator $\hat{\theta}^{(d)}$ is asymptotically normal, i.e.,
  $$\sqrt{d} (\hat{\theta}^{(d)} - \theta^*) \overset{\mathrm{dist.}}{\longrightarrow} \mathcal{N}\big(0, \tilde{\Sigma}\big),$$
  where $\tilde{\Sigma} = H^{-1}(\theta^*)\Sigma H^{-1}(\theta^*)$.
\end{thm}

Here $\,\overset{\mathrm{p.}}{\longrightarrow}\,$ and $\,\overset{\mathrm{dist.}}{\longrightarrow}\,$ denote convergence in probability and distribution, respectively.
A proof can be found in Ref.~\cite{newey_large_nodate}.
We are now equipped to provide a rigorous statement and proof of Theorem \ref{thm:error-informal} in the main text.

\begin{thm}[Asymptotic error]
  Let $\hat{\mathcal{L}}^{(d)}$ and $\hat{\theta}^{(d)}$ be defined as above.
  Let $\Theta \subset \rr^\nu$ be compact and suppose the true parameters $\theta^*$ are contained in the interior of $\Theta$.
  Suppose the class
  $$\mathcal{C} = \lbrace H(\theta) \vert \theta \in \Theta \rbrace$$
  of parametrized Hamiltonians is well conditioned, such that the estimator $\hat{\theta}^{(d)}$ is consistent and the Hessian at the true parameter value is non-singular.
  Furthermore, let the parametrization $\theta \mapsto H(\theta)$ be twice continuously differentiable.
  Then for any $\delta \in (0,1]$ there exists a function $f(d) = \mathcal{O}(d^{-1/2})$ and some value $D\in\mathbb{N}$ such that $\pp\big[\hat{\epsilon}^{(d)} > f(d)\big] < \delta$ for all $d \geq D$.
\end{thm}
\begin{proof}
  First we verify the conditions in Theorem \ref{thm:normality} are satisfied.
  By assumption condition \emph{1}, \emph{2} and \emph{6} are satisfied.
  By Lemma \ref{lem:derivatives} and \ref{lem:zero-mean} conditions \emph{3} and \emph{4} are satisfied, respectively.
  Furthermore, both $\loss$ and $\hat{\loss}^{(d)}$ are twice continuously differentiable and by the law of large numbers $\hat{\loss}^{(d)}(\theta) \overset{\mathrm{p.}}{\longrightarrow} \loss(\theta)$ for all $\theta \in \Theta$.
  Hence, we also have $\nabla^2_\theta\hat{\loss}^{(d)}(\theta)  \overset{\mathrm{p.}}{\longrightarrow} \nabla^2\loss(\theta)$, element-wise in the interior of $\Theta$, which suffices to show that condition \emph{5} is satisfied, where $\nabla^2\loss$ is the required function $H$.

  Now denote a Euclidean ball of radius $r$ centered at $0$ in $\rr^\nu$ by $B_r$ and its complement by $\bar{B}_r = \rr^\nu \setminus B_r$.
  Furthermore, let $Z$ be a random variable with distribution $\mathcal{N}(0,\tilde{\Sigma})$ and denote its probability density function by $p$.
  Now define
  \begin{align}
    r(\delta) := \min \lbrace r\vert \mathbb{P}(Z \in \bar{B}_r) \leq \delta \rbrace.
  \end{align}
  Since $p$ is normalized we have $r(\delta) < \infty$ for all $\delta \in (0, 1]$.
  Now note that
  \begin{align}
      \mathbb{P}\big[\sqrt{d}\,Z \in \bar{B}_{r(\delta)}\big]
      = \mathbb{P}\big[Z \in \bar{B}_{r(\delta)/ \sqrt{d}}\big],
  \end{align}
  which can be verified by writing out the integral over $p$ and scaling the argument of $p$ by a factor of $1/\sqrt{d}$.
  Together with this identity and the fact that $\sqrt{d}(\hat{\theta}^{(d)} - \theta^*)$ converges in distribution to $Z$, we can conclude that
  \begin{align}
      \underset{d\rightarrow \infty}{\lim} a_d = \delta,
  \end{align}
  where $a_d(\delta) := \mathbb{P}\big[\hat{\theta}^{(d)}- \theta^* \in \bar{B}_{r(\delta)/ \sqrt{d}}\big]$.
  Now consider the sequence $a_d(\delta/2)$.
  Since it converges to $\delta / 2$ there exists a value $D\in\mathbb{N}$ such that $a_d(\delta / 2) < \delta$ for all $d \geq D$.
  Lastly, the condition $\hat{\theta}^{(d)} - \theta^* \in \bar{B}_{r(\delta/2)/\sqrt{d}}$ translates to 
  \begin{align}
      \hat{\epsilon}^{(d)} < \frac{r(\delta/2)}{\sqrt{d}\,\Vert \theta^* \Vert_2},
  \end{align}
  which gives us the function $f$ described in the theorem.
\end{proof}

Note that showing that the class of parametrized 
Hamiltonians $\mathcal{C}$ is well conditioned, in the sense of the theorem, is not trivial.
Additional restrictions are necessary to exclude pathological cases in which the times $\lbrace t_j\rbrace$ are commensurate with the oscillations in the state amplitudes during time evolution.
Also one must ensure that the initial state has non-zero overlap with sufficiently many eigenstates of the Hamiltonian to be able to observe meaningful dynamics in the system.
In practice one is unlikely to encounter these issues, which is also exemplified by the fact that one can, in principle, recover a generic geometrically local and traceless Hamiltonian from only one time stamp using a generic initial state \cite{li_hamiltonian_2020}.

\section{Loss function}
\label{sec:landscape}
The analysis of our method is based on the relative error $\epsilon$.
However, in a practical setting, where one doesn't know the Hamiltonian, the error is not accessible.
Instead one needs to gauge convergence and success of the optimization by the loss function $\mathcal{L}$.
To demonstrate that one can indeed distinguish outliers, as shown in Fig.~\ref{fig:meta-params-panel} from successfully converged parameters the error is plotted against the loss function value in Fig.~\ref{fig:loss-vs-error}.

\begin{figure}
  \centering
  \includegraphics[width=.48\textwidth]{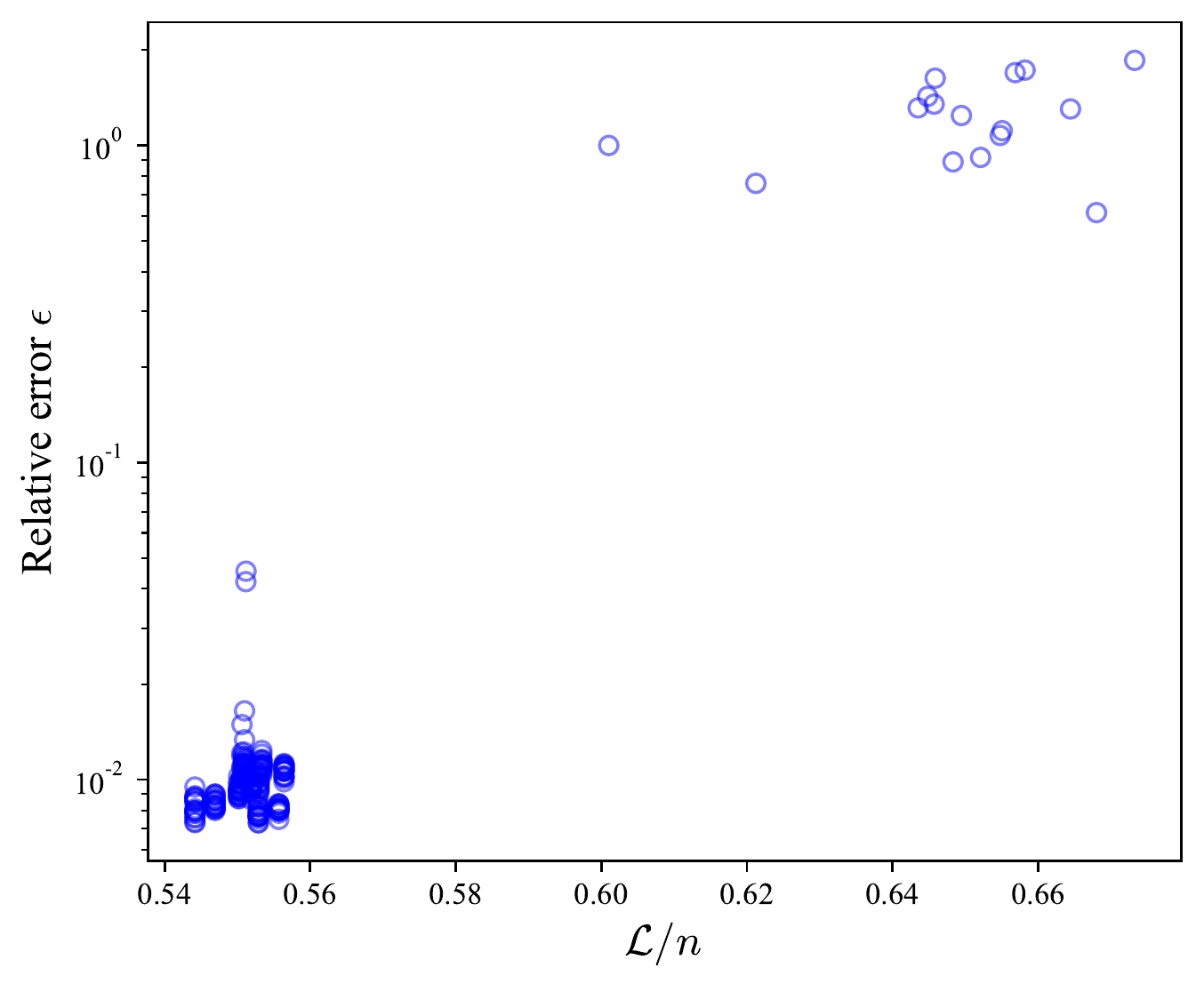}
  \caption{
    The same data points as in Fig.~\ref{fig:meta-params-panel}(b) are shown.
    Here the relative error $\epsilon$ (for various system sizes) is plotted against the loss function value $\loss$, normalized by the system size $n$.
    One can clearly divide outliers from successfully converged parameter values purely based on the loss function value.
  }
  \label{fig:loss-vs-error}
\end{figure}

To investigate the impact various experimental and meta parameters have on the loss function landscape and the ability to recover the true Hamiltonian, we plotted the landscape in a model with $\nu=2$ parameters.
In this example instead of sampling Pauli bases at random, only the bases $X^{\otimes n}$, 
$Y^{\otimes n}$, and 
$Z^{\otimes n}$ have been 
considered.
The results are shown in Fig.~\ref{fig:meta-params-panel}.
In \textbf{(b)} the number of bit-strings has been reduced. In \textbf{(c)}, only bit-strings from all-$Z$ measurements have been used.
In \textbf{(d)}, 
only samples from the first two time stamps have been made use of.
In \textbf{(e)}, 
the Trotter step size has been increased leading to higher errors.
In \textbf{(f)}, 
the bond dimension has been reduced, again increasing the 
errors.
It becomes apparent, that only measuring in the $Z$ basis has negative effects on the loss landscape.
Also one can clearly see that too short times are not sufficient for the loss function to form a well pronounced minimum.

\begin{figure}
  \centering
  \includegraphics[width=.48\textwidth]{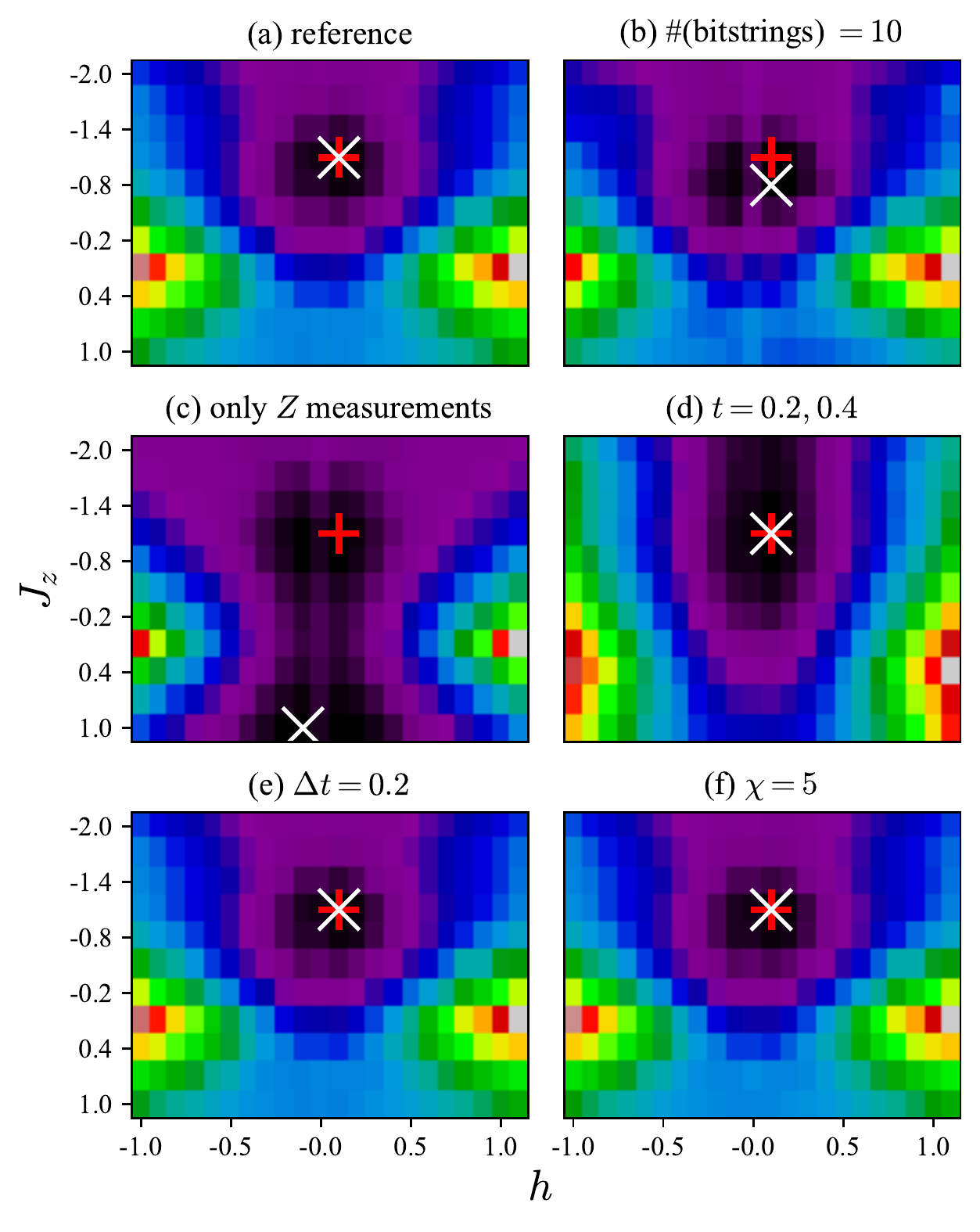}
  \caption{The loss function in the parameters $h$ (uniform over all spins) and $J_z$ in arbitrary units (black is low).
  To compute the loss 1000 bit-strings have been sampled for each basis choice, all-$X$, all-$Y$, all-$Z$, and for each time stamp $t = 0.2, 0.4, 0.6, 0.8, 1.0$.
  The Trotter step size is $\Delta t = 0.05$ and the bond dimension is $\chi = 10$.
  The red cross shows the true parameters, while the white cross shows the minimum of the loss function.
  For reference the loss landscape is shown in \textbf{(a)} with these exact meta parameters.}
  \label{fig:meta-params-panel}
\end{figure}

\section{Differentiation of the singular value decomposition}
\label{sec:svd-ad}

The \emph{singular value decomposition} 
(SVD) is an integral part of many tensor network 
algorithms \cite{schollwock_density-matrix_2011}.
Apart from linear algebra applications such as computing the rank or the kernel of a matrix, the SVD is widely used in physics, for instance to access the Schmidt spectrum, the entanglement entropy, and reduced density matrices of bipartite quantum states.

The Jacobian-vector product can be used for forward-mode automatic differentiation, but allows reverse-mode automatic differentiation as well by using transposition rules, thereby avoiding the need to work with the dual transformation (vector-Jacobian product) explicitly.
As a design choice, the popular open-source automatic differentiation framework JAX \cite{jax2018github} contains only Jacobian-vector product rules for non-linear transformations (such as the SVD) along with transposition rules for linear transformations.
The results derived here have been included into JAX ({\small\href{https://github.com/google/jax/pull/5225}{\texttt{github.com/google/jax/pull/5225}}}). \smallskip

\textbf{Previous work.} We emphasize that the vector-Jacobian product of the complex SVD has been worked out previously \cite{wan_automatic_2019,liu_linear_2019}.
These results provide the basis for our derivation of the JVP rule.
Our aim is to make the differentiation of the SVD accessible to a broad audience, not necessarily familiar with the field of automatic differentiation.
Additionally, 
we provide the Jacobian-vector product and discuss its lack of uniqueness.
The reader may also be referred to a larger collection of automatic differentiation formulas for linear algebra transformations \cite{giles_extended_nodate}. \smallskip

\textbf{Notation.} Let $\I$ be the identity matrix and $\bar{\I}$ its complement, i.e., where the diagonal is zero and all off-diagonal entries are one.
If necessary we will annotate the matrix dimensions, e.g., $\I_{n\times m}$.
Let $\circ$ denote the element-wise matrix product, say, $(A\circ B)_{i,j} = A_{i,j} B_{i,j}$. \smallskip

\textbf{Automatic differentiation.} Imagine we have a function $L:\mathbb{R}\rightarrow\mathbb{R}$ which is composed of several functions $L(p) = f(\mathsf{S}(g(p)))$, where $\mathsf{S}$ may denote the SVD.
The derivative of $L$ is then given by the chain rule $\mathrm{D}L = \mathrm{D}f\cdot\mathrm{D}\mathsf{S}\cdot\mathrm{D}g$.
To compute the derivative, we can multiply these three Jacobians beginning from the left or the right, which is called backward-mode or forward-mode differentiation, respectively.
In the former case one needs the vector-Jacobian product, in the latter the Jacobian-vector product.
This technique can be used on more complex functions $L$ as well.
Instead of elaborating on more general cases we refer the reader to Ref.~\cite{bartholomew-biggs_automatic_2000}. \smallskip 

\textbf{Derivation.} We begin by stating the relation between the input $A \in \C^{n \times m}$ of the SVD and its output $U \in \C^{n \times k}, S = \diag(s_1, \ldots, s_k) \in \C^{k \times k}, \dg{V} \in \C^{k \times m}$, where $s_i\in\mathbb{R}$ for all $i$ and $k = \mathrm{min}(n, m)$.
We have
\begin{equation}
  A = U S \dg{V}.
\end{equation}
The total derivative of this relation \cite{total_derivative_footnote}
\begin{equation}
  \dd A = \dd U S \dg{V} + U \dd S \dg{V} + U S \dd\dg{V}
\end{equation}
allows us to derive the differentiation rules.
As we will see, this derivative is not well defined as it stands, since there is a gauge freedom in the matrices $U$ and $\dg{V}$.
Our aim is to find expressions for $\dd U$, $\dd S$, and $\dd\dg{V}$ in terms of $\dd A$. \smallskip

\textbf{Square matrices.} We assume $m = n$ for the sake of lucidity and explain additional steps necessary for non-square matrices later.
First we transform $\dd A$ by multiplying $\dg{U}$ and $V$ from the right and left, respectively.
We obtain
\begin{equation}
  \dd \tilde{A} = \dd \tilde{U} S + \dd S + S \dd\dg{\tilde{V}},
\end{equation}
where $\dd \tilde{U} = \dg{U}\dd U$ and $\dd \tilde{V} = \dg{V} \dd{V}$.
Note that
\begin{equation}
  \label{eq:skew-symmetry}
  \dd \tilde{U} = -\dd \dg{\tilde{U}} \quad \dd \tilde{V} = -\dd \dg{\tilde{V}}
\end{equation}
due to $U$ and $V$ being unitary, i.e., $\dd (\dg{U}U) = \dd \I = 0$.
This implies that the diagonal of $\dd \tilde{U}$ and $\dd \tilde{V}$ are purely imaginary and we can eliminate them by adding $\dd \tilde{A}$ to its Hermitian conjugate.
With this we can already state the derivative of the singular values
\begin{equation}
  \label{eq:dS}
  \dd S = \frac{1}{2}(\dd \tilde{A} + \dd \dg{\tilde{A}}).
\end{equation}
To compute $\dd \tilde{U}$ and $\dd \tilde{V}$ we first focus on the off-diagonal portion.
This is closely related to the derivation for SVD of real matrices \cite{townsend_differentiating_2016}.
In contrast, the diagonals need extra care, which has no counterpart in the real case.
Let $i\neq j$ such that we can disregard $\dd S$ and consider the matrix 
element
\begin{align}
  (\dd \tilde{A}S + S \dd\dg{\tilde{A}})_{i,j} =~& \dd \tilde{U}_{i,j} S_{j,j}^2 + S_{i,i}\dd \dg{\tilde{V}}_{i,j}S_{j,j} \nonumber\\
  &+ S_{i,i}^2 \dd \dg{\tilde{U}}_{i,j}  + S_{i,i}\dd \tilde{V}_{i,j}S_{j,j},
\end{align}
using the Einstein summation convention.
Using (\ref{eq:skew-symmetry}), 
we see that the two terms involving $\dd\tilde{V}$ vanish while the other two can be combined.
We can solve the equation element-wise for $\dd \tilde{U}$ using the matrix with entries $F_{i,j} = 1/(s_j^2 - s_i^2)$ for $i\neq j$ and zero otherwise.
We can use the same matrix to solve for $\dd \dg{\tilde{V}}$ and obtain
\begin{align}
  \label{eq:dtildeU}
  \bar{\I} \circ \dd\tilde{U} &= F \circ \frac{1}{2}(\dd \tilde{A}S + S \dd\dg{\tilde{A}}) ,\\
  \label{eq:dtildeV}
  \bar{\I} \circ \dd\tilde{V} &= F \circ \frac{1}{2}(S\dd\tilde{A} + \dd\dg{\tilde{A}}S).
\end{align}
What is left is to determine the diagonals.
We do this by considering the diagonal of $\dd \tilde{A} - \dd \dg{\tilde{A}}$.
Since this eliminates the (purely real) components of $\dd S$ we are left with
\begin{equation}
  \label{eq:dUdV_diag}
  S^{-1} \circ \frac{1}{2}(\dd \tilde{A} - \dd \dg{\tilde{A}}) = \I \circ (\dd\tilde{U} + \dd\dg{\tilde{V}}).
\end{equation}
This tells us the sum of $\dd\tilde{U}$'s and $\dd\dg{\tilde{V}}$'s diagonals, but not the individual summands.
We will now show that we can distribute this sum in any arbitrary way (for instance half-half) between the two matrices.
To see this consider an arbitrary matrix element of $\dd A$
\begin{align}
  \label{eq:dA_il_1}
  \dd A_{i,l} 
  =&~ U_{i,j}\dd\tilde{U}_{j,k}S_{k,k}V^*_{l,k} \\
  \label{eq:dA_il_2}
  &+ U_{i,j}S_{j,j}\dd \tilde{V}^*_{k,j} V^*_{l,k} \\
  &+ U_{i,j}\dd S_{j,j}V^*_{l,j}. \nonumber
\end{align}
We are interested in all summands in (\ref{eq:dA_il_1}) and (\ref{eq:dA_il_2}) that contain diagonal elements of $\dd\tilde{U}$ and $\dd\dg{\tilde{V}}$, i.e., 
summands where $j = k$.
This reduced sum turns out to be $U_{i,j}S_{j,j}V^*_{l,j}(\dd\tilde{U}_{j,j} + \dd\tilde{V}^*_{j,j})$
which implies the above claim, that any solution of equation (\ref{eq:dUdV_diag}) suffices to recover $\dd A$.
Hence one valid Jacobian-vector product rule is
\begin{align}
  \label{eq:dU}
  \dd U &= \frac{U}{2} \big[F \circ (\dd \tilde{A}S + S \dd\dg{\tilde{A}}) + S^{-1} \circ (\dd \tilde{A} - \dd \dg{\tilde{A}}) \big],\\
  \label{eq:dV}
  \dd V &= \frac{V}{2} \big[ F \circ (S\dd\tilde{A} + \dd\dg{\tilde{A}}S) \big].
\end{align} \\

\textbf{Gauge freedom.}
In addition to the freedom of choice of the diagonal of $\dd \tilde{U}$ and $\dd\tilde{V}$ there is a freedom in the choice of $U$ and $V$ in the SVD.
Notice that the SVD is invariant under the transformation $U \mapsto U\Lambda$ and $\dg{V} \mapsto \dg{\Lambda}\dg{V}$, with $\Lambda = \mathrm{diag}(\ee^{\ii\phi_1}, \ldots, \ee^{\ii\phi_k})$ for arbitrary $\phi_i$'s as $\dg{\Lambda}S\Lambda = S$.
Consequently the derivatives transform as $\dd U \mapsto \dd U\Lambda$ and $\dd\dg{V} \mapsto \dg{\Lambda}\dd\dg{V}$.
Therefore, we can only consider functions $L$ (discussed in the beginning) which, at some point during the calculation, eliminate the gauge freedom, perhaps by multiplying $U$ with $f(S)$ and then $\dg{V}$ for some differentiable function $f$. \smallskip

\textbf{Non-square matrices.}
To solve the case where $n\neq m$ we assume that $n>m$, such that $U\in \mathbb{C}^{n\times m}$, $S\in\mathbb{R}^{m\times m}$, and $V \in \mathbb{C}^{m\times m}$.
It will be straightforward to repeat the following arguments for $n<m$.
Following the reasoning presented in Ref.~\cite{townsend_differentiating_2016} we first notice that $U$ is no longer an element of the unitary group. In fact, it is an element of the \emph{Stiefel manifold} $V_m(\mathbb{C}^{n})$.
While it is still true that $\dg{U}U = \I_{m\times m}$, we now have $U\dg{U} \neq \I_{n\times n}$.
It follows that $\dd \tilde{U} \in \mathbb{C}^{m\times m}$ is still well defined and anti-Hermitian $\dd \tilde{U} = -\dd\dg{\tilde{U}}$.
Since these are the properties we have used to derive the expressions for $\dd S$, $\dd\tilde{U}$ and $\dd\tilde{V}$ (Eqs.~(\ref{eq:dS}), (\ref{eq:dtildeU}), and (\ref{eq:dtildeV})) they remain valid in the non-square case.

However, extra care is required when recovering $\dd U$ since we cannot simply multiply $\dd\tilde{U}$ by $U$ from the left.
Instead we need to extend $U \in \mathbb{C}^{n\times m}$ by $U_\perp \in \mathbb{C}^{n\times(n-m)}$ such that the stacked matrix $[U, U_\perp] \in \mathbb{C}^{n\times n}$ is unitary.
Here the columns of $U_\perp$ consist of orthonormal vectors in the complement of the image of $U$.
This allows us to recover the $n\times n$ identity
\begin{equation}
  \label{eq:nbyn_identity}
  [U, U_\perp] \cdot
  \begin{bmatrix}\dg{U}\\\dg{U}_\perp\end{bmatrix}
  = U\dg{U} + U_\perp\dg{U}_\perp
  = \I_{n\times n}.
\end{equation}
Analogously, 
we extend $\dd\tilde{U}$ by an element $\dd\tilde{U}_\perp \in \mathbb{C}^{(n-m) \times m}$ which acts in the image of $U_\perp$ such that
\begin{equation}
  \label{eq:dU_recovery}
  \dd U = U\dd\tilde{U} + U_\perp\dd\tilde{U}_\perp.
\end{equation}
This allows us to compute $\dd\tilde{U}_\perp$ by multiplying $\dd A$ by $\dg{U}_\perp$ which yields $\dd\tilde{U}_\perp = \dg{U}_\perp \dd A V S^{-1}$.
Plugging this into (\ref{eq:dU_recovery}) we have the recovery in terms of known matrices
\begin{equation}
  \label{eq:dU_recovery2}
  \dd U = U\dd\tilde{U} + (\I - U\dg{U}) \dd A V S^{-1},
\end{equation}
where we have used (\ref{eq:nbyn_identity}) to substitute $U_\perp \dg{U}_\perp$ and as before we may choose
\begin{equation}
  \dd\tilde{U} = \frac{1}{2} \big[F \circ (\dd \tilde{A}S + S \dd\dg{\tilde{A}}) + S^{-1} \circ (\dd \tilde{A} - \dd \dg{\tilde{A}}) \big].
\end{equation}
Analogously we can derive the Jacobian-vector product rule for $n<m$ and obtain
\begin{equation}
  \label{eq:dV_recovery}
  \dd\dg{V} = \dd\dg{\tilde{V}}\dg{V} + S^{-1}\dg{U}\dd A (\I - V\dg{V}).
\end{equation}
\smallskip

\textbf{Choice of the diagonal component.}
As shown in (\ref{eq:dUdV_diag}) and subsequent arguments one has the freedom to split the diagonal given in (\ref{eq:dUdV_diag}) and add its parts to $\dd \tilde{U}$ and $\dd\tilde{\dg{V}}$ given in (\ref{eq:dtildeU}) and (\ref{eq:dtildeV}).
In order to demonstrate the effect of this splitting on the precision of the Jacobian-vector product rule we differentiated the function
\begin{align}
  \label{eq:test_function}
  f(p) &= \mathrm{Re}(U(p)\cdot S(p)\cdot \dg{V}(p)), \\
  U(p), S(p), \dg{V}(p) &= \mathsf{SVD}(\exp(A + pB)) \nonumber
\end{align}
for 100 pairs of random $100\times100$ matrices $A$ and $B$, where the exponential function is applied \emph{element-wise}.
A splitting parameter $\alpha$ is multiplied to the diagonal (Eq.~\ref{eq:dUdV_diag}) which is then added to $\dd \tilde{U}$ (Eq.~\ref{eq:dtildeU}) while a $(1-\alpha)$ multiple of the diagonal is added to $\dd\dg{\tilde{V}}$ 
(Eq.~\ref{eq:dtildeV}).
The differential of the function (\ref{eq:test_function}) is evaluated at $p=20$ and compared to the analytical derivative $\dd f/\dd p = \exp(A + pB)\circ B$.
\begin{figure}
  \includegraphics[width=.53\textwidth]{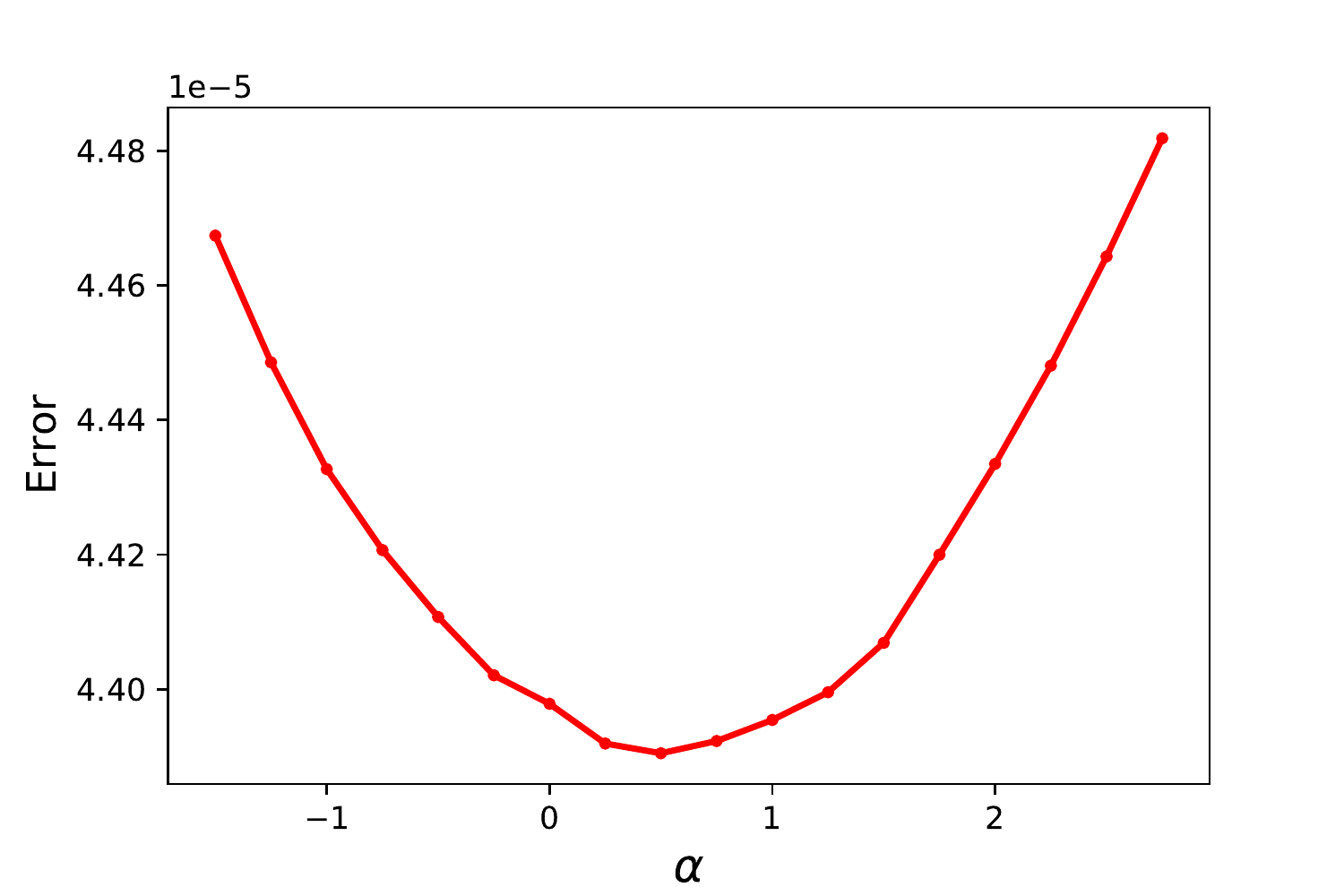}
  \caption{\label{fig:test}The error for 100 pairs of random matrices $A$ and $B$ for various values of $\alpha$. The error is defined by the Hilbert-Schmidt distance between the analytical derivative and the automatic-differentiation result $\Vert \dd f/\dd p - \mathsf{autodiff}(f) \Vert_\mathrm{HS}$.}
\end{figure}
The results shown in 
Fig.~\ref{fig:test} suggest that an even splitting, i.e., $\alpha=1/2$ yields the highest precision of the Jacobian-vector product rule. However, it might be more efficient to choose $\alpha = 0$ or $\alpha = 1$ to reduce computational operations. 

\end{document}